 \documentclass[a4paper,UKenglish,cleveref,autoref,thm-restate]{lipics-v2021}
\nolinenumbers
\hideLIPIcs

\usepackage{amsmath}
\usepackage{mathtools}
\usepackage{amsthm}
\usepackage{thmtools, thm-restate}
\usepackage{physics}    
\usepackage{amssymb}
\usepackage{amsfonts}
\usepackage{mathdots}   
\usepackage{mathrsfs}   
\usepackage{nicefrac}   
\usepackage{stmaryrd}
\usepackage{tikz-cd}     
\usepackage[all,cmtip]{xy}   
\usepackage{arydshln}  

\usepackage{enumitem}

\usepackage{hyphenat}

\usepackage{fix-cm}
\usepackage{anyfontsize}

\usepackage[dvipsnames]{xcolor}

\usepackage{xspace}

\usepackage{tikz}
\usetikzlibrary{decorations.markings}
\usepackage{tikzit}
\usepackage{float}
\usetikzlibrary{shadows,fadings}
\input{styles/symplectic-ZX.tikzdefs}

\tikzstyle{wirelable}=[font={\scriptsize}, scale=.9, fill=none, inner sep=1pt, tikzit category=strings]
\tikzstyle{box}=[fill=white, draw=black, shape=rectangle, inner sep=2.5pt, font={\scriptsize}, tikzit category=strings]
\tikzstyle{thick_box}=[fill=white, line width=.8pt, tikzit draw=blue, draw=black, shape=rectangle, inner sep=2.5pt, font={\scriptsize}, tikzit category=strings]

\tikzstyle{background}=[-, fill={rgb,255: red,220; green,220; blue,220}, draw=none, tikzit draw={rgb,255: red,128; green,128; blue,128}, tikzit category=strings]
\tikzstyle{bbox}=[-, fill=white, tikzit category=strings]
\tikzstyle{thick_bbox}=[-, fill=white, line width=.8pt, tikzit draw=blue, tikzit category=strings]
\tikzstyle{thick}=[-, line width=.8pt, tikzit draw=blue, tikzit category=strings]
\tikzstyle{dash_edge}=[-, dashed]

\tikzstyle{control}=[thick, fill=black, circle, scale=1, inner sep=.05cm, tikzit category=circuits]
\tikzstyle{target}=[inner sep=2.5pt, draw, circle, path picture={ \draw[black](path picture bounding box.east) -- (path picture bounding box.west) (path picture bounding box.south) -- (path picture bounding box.north);}, tikzit category=circuits]
\tikzstyle{gate}=[fill=white, draw=black, shape=rectangle, tikzit category=circuits]

\tikzstyle{bgate}=[-, fill=white, tikzit category=circuits]

\tikzstyle{wn}=[font={\scriptsize\boldmath}, inner sep=1mm, outer sep=-1.8mm, scale=0.8, tikzit shape=circle, draw=black, fill=black!01, tikzit fill=white, tikzit draw=black, shape=circle, tikzit category=GLA]
\tikzstyle{bn}=[font={\scriptsize\boldmath}, inner sep=1mm, outer sep=-1.8mm, scale=0.8, tikzit shape=circle, draw=black, fill={rgb,255: red,100; green,100; blue,100}, tikzit draw=black, shape=circle, tikzit category=GLA]
\tikzstyle{gwn}=[font={\scriptsize\boldmath}, inner sep=1mm, outer sep=-1.8mm, scale=0.8, tikzit shape=circle, draw=black, fill=black!01, tikzit draw=blue, tikzit fill=white, tikzit category=ZX, shape=circle, line width=.8pt]
\tikzstyle{gbn}=[font={\scriptsize\boldmath}, inner sep=1mm, outer sep=-1.8mm, scale=0.8, tikzit shape=circle, draw=black, fill={rgb,255: red,100; green,100; blue,100}, tikzit fill=black, tikzit draw=blue, shape=circle, tikzit category=ZX, line width=.8pt]

\tikzstyle{wphase}=[rounded rectangle, rounded rectangle arc length=120, fill={white}, inner sep=2pt, font={\tiny\boldmath}, label distance=1mm, fill opacity=.8, text opacity=1, tikzit category=GLA, tikzit fill=white, tikzit draw=black]
\tikzstyle{bphase}=[rounded rectangle, rounded rectangle arc length=120, fill={rgb,255: red,100; green,100; blue,100}, inner sep=2pt, font={\tiny\boldmath}, label distance=1mm, fill opacity=.6, text opacity=1, tikzit category=GLA, tikzit fill=gray, tikzit draw=gray]
\tikzstyle{mphase}=[rounded rectangle, rounded rectangle arc length=120, fill={rgb,255: red,180; green,180; blue,180}, inner sep=2pt, font={\tiny\boldmath}, label distance=1mm, fill opacity=.6, text opacity=1, tikzit category=GLA]

\tikzstyle{lmult}=[shape=signal, signal to=west, signal from=east, fill={rgb,255: red,180; green,180; blue,180}, draw=black, minimum height=6pt, inner sep=1pt, font={\scriptsize\boldmath}, tikzit fill=gray, tikzit category=GLA, anchor=center, outer sep=-.1cm, signal pointer angle=\arrowangle,scale=\arrowscaling]
\tikzstyle{rmult}=[shape=signal, signal to=east, signal from=west, fill={rgb,255: red,180; green,180; blue,180}, draw=black, minimum height=6pt, inner sep=1pt, font={\scriptsize\boldmath}, tikzit fill=gray, tikzit category=GLA, anchor=center, outer sep=-.1cm, signal pointer angle=\arrowangle,scale=\arrowscaling]
\tikzstyle{dmult}=[shape=signal, signal to=east, signal from=west, fill={rgb,255: red,180; green,180; blue,180}, draw=black, minimum height=6pt, inner sep=1pt, font={\scriptsize\boldmath}, tikzit fill=gray, tikzit category=GLA, rotate=270, anchor=center, outer sep=-.1cm, signal pointer angle=\arrowangle,scale=\arrowscaling]
\tikzstyle{umult}=[shape=signal, signal to=east, signal from=west, fill={rgb,255: red,180; green,180; blue,180}, draw=black, minimum height=6pt, inner sep=1pt, font={\scriptsize\boldmath}, tikzit fill=gray, tikzit category=GLA, rotate=90, anchor=center, outer sep=-.1cm, signal pointer angle=\arrowangle,scale=\arrowscaling]

\tikzstyle{lmat}=[shape=signal, signal to=west, signal from=east, fill={zx_grey}, draw=black, minimum height=6pt, inner sep=1pt, font={\scriptsize\boldmath}, tikzit fill=gray, tikzit category=GLA, anchor=center, outer sep=-.1cm, signal pointer angle=\arrowangle,scale=\arrowscaling, line width=.8pt, tikzit draw=blue]
\tikzstyle{rmat}=[shape=signal, signal to=east, signal from=west, fill={zx_grey}, draw=black, minimum height=6pt, inner sep=1pt, font={\scriptsize\boldmath}, tikzit fill=gray, tikzit category=GLA, anchor=center, outer sep=-.1cm, signal pointer angle=\arrowangle,scale=\arrowscaling, line width=.8pt, tikzit draw=blue]
\tikzstyle{dmat}=[shape=signal, signal to=east, signal from=west, fill={zx_grey}, draw=black, minimum height=6pt, inner sep=1pt, font={\scriptsize\boldmath}, tikzit fill=gray, tikzit category=GLA, rotate=270, anchor=center, outer sep=-.1cm, signal pointer angle=\arrowangle,scale=\arrowscaling, line width=.8pt, tikzit draw=blue]
\tikzstyle{umat}=[shape=signal, signal to=east, signal from=west, fill={zx_grey}, draw=black, minimum height=6pt, inner sep=1pt, font={\scriptsize\boldmath}, tikzit fill=gray, tikzit category=GLA, rotate=90, anchor=center, outer sep=-.1cm, signal pointer angle=\arrowangle,scale=\arrowscaling, line width=.8pt, tikzit draw=blue]

\tikzstyle{gn}=[font={\scriptsize\boldmath}, inner sep=1mm, outer sep=-1.8mm, scale=0.8, tikzit shape=circle, draw=black, fill={rgb,255: red,216; green,248; blue,216}, tikzit draw=black, tikzit fill=green, tikzit category=GSA, shape=circle]
\tikzstyle{rn}=[font={\scriptsize\boldmath}, inner sep=1mm, outer sep=-1.8mm, scale=0.8, tikzit shape=circle, draw=black, fill={rgb,255: red,232; green,165; blue,165}, tikzit fill=red, tikzit draw=black, shape=circle, tikzit category=GSA]
\tikzstyle{ggn}=[font={\scriptsize\boldmath}, inner sep=1mm, outer sep=-1.8mm, scale=0.8, tikzit shape=circle, draw=black, fill={zx_green}, tikzit draw=blue, tikzit fill=green, tikzit category=ZX, shape=circle, line width=.8pt]
\tikzstyle{grn}=[font={\scriptsize\boldmath}, inner sep=1mm, outer sep=-1.8mm, scale=0.8, tikzit shape=circle, draw=black, fill={zx_red}, tikzit fill=red, tikzit draw=blue, shape=circle, tikzit category=ZX, line width=.8pt]

\tikzstyle{had}=[fill=yellow, draw=black, shape=rectangle, tikzit category=GSA, tikzit fill=yellow, tikzit draw=black, inner sep=1.5pt, minimum height = 5pt, minimum width = 5pt, font={\scriptsize\boldmath}]
\tikzstyle{ghad}=[shading=hadballshading, fill={zx_grey}, draw=black, shape=rectangle, tikzit category=GSA, tikzit fill=yellow, tikzit draw=blue, thick, minimum size=5pt, inner sep=1.5pt, scale={\boxscaling}, font={\scriptsize\boldmath}]

\tikzstyle{gphase}=[rounded rectangle, rounded rectangle arc length=120, fill={zx_green}, inner sep=2pt, font={\tiny\boldmath}, label distance=1mm, fill opacity=.8, text opacity=1, tikzit category=GSA, tikzit fill=green, tikzit draw=green]
\tikzstyle{rphase}=[rounded rectangle, rounded rectangle arc length=120, fill={zx_red}, inner sep=2pt, font={\tiny\boldmath}, label distance=1mm, fill opacity=.6, text opacity=1, tikzit category=GSA, tikzit fill=red, tikzit draw=red]

\tikzstyle{gbox}=[fill=white, line width=1pt, draw=black, shape=rectangle, inner sep=2.5pt, font={\scriptsize}, tikzit draw=blue]

\tikzstyle{graph_vertex}=[fill=black, draw=black, shape=circle, tikzit category=mbqc, minimum size=2.4mm, inner sep=.8mm]
\tikzstyle{graph_weight}=[fill=white, draw=none, shape=rectangle, tikzit category=mbqc, inner sep=2pt, scale=.8]

\usepackage{mathpartir}

\ifx\supresscomments\undefined%
  \newcommand{\robert}[1]{%
    \noindent{\color{purple} \textsf{[RIB: #1]}}
  }
  \ifx\supresscolecomments\undefined%
    \newcommand{\cole}[1]{%
      \noindent{\color{blue} \textsf{[CC: #1]}}
    }
  \else%
    \newcommand{\cole}[1]{}
  \fi
\else%
  \newcommand{\robert}[1]{}
  \newcommand{\cole}[1]{}
\fi

\newcommand{\Aff}{%
  \mathsf{Aff}
}
\newcommand{\Lin}{%
  \mathsf{Lin}
}
\newcommand{\Lag}{%
  \mathsf{Lag}
}

\newcommand{\Isot}{%
  \mathsf{Isot}
}

\newcommand{\zx}{%
  \mathsf{GSA}
}

\newcommand{\gla}{%
  \mathsf{GLA}
}
\newcommand{\gaa}{%
  \mathsf{GAA}
}

\newcommand{\Proj}{%
  \operatorname{\mathbb{P}}
}

\renewcommand{\op}{\mathsf{op}}

\newcommand{\Conf}[2]{\bigl[\,#1 \bigm| #2\,\bigr]}
\newcommand{\aden}[3]{\bigl\langle #1\bigr\rangle_{#3}^{#2}} 

\newcommand{\interp}[1]{%
  \left\llbracket #1 \right\rrbracket
}
\newcommand{\trans}{%
  \mathsf{T}
}
\DeclareRobustCommand{\disc}{
  {!_{\sf isom}}
}

\newcommand{\N}{%
  \mathbb{N}
}
\newcommand{\Z}{%
  \mathbb{Z}
}
\newcommand{\Zp}{%
  \mathbb{F}_p
}
\newcommand{\K}{%
  \mathbb{K}
}

\newcommand{\C}{%
  \mathbb{C}
}

\newcommand{\F}{%
  \mathbb{F}
}
\NewDocumentCommand{\Fpnn}{O{n}}{\F_p^{2(#1)}}

\NewDocumentCommand{\Mat}{ O{m} O{n} O{\F_p} }{%
  \operatorname{M}_{#2 \times #1}(#3)
}

\newcommand{\FinRel}{\mathsf{FRel}}

\newcommand{\FHilb}{\mathsf{FHilb}}

\newcommand{\Stab}{\mathsf{Stab}}
\NewDocumentCommand{\Hp}{O{n}}{\mathcal{H}_p^{\otimes #1}}
\NewDocumentCommand{\HS}{O{S}}{\mathcal{H}_{#1}}
\NewDocumentCommand{\Pauli}{O{n}}{\mathcal{P}_p^{\otimes #1}}
\NewDocumentCommand{\Cliff}{O{n}}{\mathcal{C}\!\ell_p^{#1}}

\DeclareMathOperator{\CPM}{CPM}
\newcommand{\Split}{\operatorname{Split}^\dagger}
\DeclareMathOperator{\Caus}{Caus}
\DeclareMathOperator{\Total}{Total}
\DeclareMathOperator{\im}{im}
\DeclareMathOperator{\rel}{Rel}

\NewDocumentCommand{\Rel}{O{X}}{%
  \mathsf{Rel}_{#1}
}
\NewDocumentCommand{\Symp}{O{\F_p}}{%
  \mathsf{Symp}_{#1}
}
\NewDocumentCommand{\ASymp}{O{\F_p}}{%
  {\Aff}\Symp[#1]
}
\NewDocumentCommand{\AR}{O{\F_p}}{%
  {\Aff}\Rel[#1]
}
\NewDocumentCommand{\lR}{O{\F_p}}{%
  {\Lin}\Rel[#1]
}
\NewDocumentCommand{\LR}{O{\F_p}}{%
  {\Lag}\Rel[#1]
}
\NewDocumentCommand{\IR}{O{\F_p}}{%
  {\Isot}\Rel[#1]
}
\NewDocumentCommand{\CR}{O{\F_p}}{%
  {\mathsf{Coisot}}\Rel[#1]
}
\NewDocumentCommand{\ALR}{O{\F_p}}{%
  {\Aff}\LR[#1]
}
\NewDocumentCommand{\AIR}{O{\F_p}}{%
  {\Aff}\IR[#1]
}
\NewDocumentCommand{\ACR}{O{\F_p}}{%
  {\Aff}\CR[#1]
}
\NewDocumentCommand{\ARQ}{O{\F_p}}{%
  {\Aff}\Rel[#1]^Q
}

\NewDocumentCommand{\NLQ}{O{\F_p}}{%
  \Rel[#1]^Q
}
\NewDocumentCommand{\ZX}{O{\F_p}}{%
  \zx_{#1}
}
\NewDocumentCommand{\ZXdisc}{O{\F_p}}{%
  \ZX[#1]^\disc
}
\NewDocumentCommand{\GLA}{O{\F_p}}{%
  \gla_{#1}  
}
\NewDocumentCommand{\GAA}{O{\F_p}}{%
  \gaa_{#1}
}

\NewDocumentCommand{\Matrices}{ O{m} O{n} O{\K} }{%
  \operatorname{Mat}_{#3}(#1,#2)
}
\NewDocumentCommand{\Sym}{O{n} O{\K} }{%
  \operatorname{Sym}_{#1}(#2)
}

\newlength\oversetwidth
\newlength\underwidth

\newcommand{\swap}{\mathtt{swap}}

\renewcommand{\phi}{%
  \varphi
}
\renewcommand{\epsilon}{%
  \varepsilon
}

\renewcommand{\leq}{%
  \leqslant
}

\renewcommand{\xi}{\chi}

\newcommand{\Ty}{\mathsf{Ty}}
\newcommand{\Reg}{\ensuremath{\mathsf{Reg}}\xspace}

\DeclareFontFamily{U}{bbold}{}
\DeclareFontShape{U}{bbold}{m}{n}{<-6>bbold5<6-8>bbold7<8->bbold10}{}

\newcommand{\fby}{%
  \mathbin{\text{\usefont{U}{bbold}{m}{n}\char"3B}}%
}

\newcommand{\pit}{\mathbf{init}}
\newcommand{\qpit}{\mathbf{qinit}}
\newcommand{\discard}{\mathbf{disc}}
\newcommand{\mul}{\mathbf{mul}}
\newcommand{\measure}{\mathbf{meas}}
\newcommand{\control}{\mathbf{ctrl}}
\newcommand{\noop}{\mathbf{skip}}
\newcommand{\muleqq}{\!\mathrel{\boldsymbol{*}\!\!\boldsymbol{=}}\!}

\newcommand{\affmul}{\!\!\mathrel{\boldsymbol{*}}\!}

\newcommand{\pittype}{\mathbf{\mathtt{pit}}}
\newcommand{\qpittype}{\mathbf{\mathtt{qpit}}}

\DeclareRobustCommand{\vbreg}[1]{%
  \underline{\mbox{\normalfont\upshape
    \usefont{T1}{pcr}{b}{n}\detokenize{#1}}}
}

\DeclareRobustCommand{\reg}[1]{%
  \underline{\mbox{\normalfont\upshape
    \usefont{T1}{pcr}{m}{n}\detokenize{#1}}}
}

\DeclareTextFontCommand{\texttt}{%
  \normalfont\upshape\fontencoding{T1}\fontfamily{pcr}\selectfont
}

\DeclareMathOperator{\dom}{dom}

\newcommand{\StabChan}{\mathsf{StabChan}}
\newcommand{\QChan}{\mathsf{QuantChan}}
\newcommand{\StabLang}{\mathrm{SPL}}
\newcommand{\nlStabLang}{\mathrm{NLSPL}}
\newcommand{\QPL}{\mathrm{QPL}}

\let\oldrightsquigarrow\rightsquigarrow
\renewcommand{\rightsquigarrow}{\mathrel{\!\oldrightsquigarrow\!}}

\DeclareMathOperator{\Gr}{Gr}

\renewcommand{\emptyset}{\varnothing}

\newcommand{\update}{\;\triangleright\;}

\renewcommand{\eta}{\mathtt{cup}}
\renewcommand{\epsilon}{\mathtt{cap}}

\title{Denotational semantics for stabiliser quantum programs}

\author{Robert I. Booth}{University of Oxford, United Kingdom}{firstname dot lastname@cs.ox.ac.uk}{https://orcid.org/0000-0002-1146-3380}{}
\author{Cole Comfort}{Universit\'e Paris-Saclay, CNRS, ENS Paris-Saclay, Inria, CentraleSup\'elec, Laboratoire M\'ethodes Formelles}{firstname dot lastname@inria.fr}{}{}
\authorrunning{R.\,I. Booth and C. Comfort}
\Copyright{Robert I. Booth and Cole Comfort} 

\ccsdesc[500]{Theory of computation~Denotational semantics}
\ccsdesc[300]{Theory of computation~Quantum computation theory}
\ccsdesc[300]{Theory of computation~Program reasoning}

\keywords{stabiliser theory, denotational semantics, quantum error-correction, symplectic linear algebra, categorical semantics, quantum programming languages} 

\category{} 

\relatedversion{} 

\acknowledgements{We thank Louis Lemonnier and Vladimir Zamdzhiev for useful discussions. This work has been partially funded by the European Union through the
MSCA SE project QCOMICAL and within the framework of ``Plan France
2030'', under the research projects EPIQ ANR-22-PETQ-0007 and HQI-R\&D
ANR-22-PNCQ-0002.}

\EventEditors{}
\EventNoEds{0}
\EventLongTitle{}
\EventShortTitle{}
\EventAcronym{}
\EventYear{}
\EventDate{}
\EventLocation{}
\EventLogo{}
\SeriesVolume{}
\ArticleNo{}

\begin{document}

\maketitle

\begin{abstract}
  The stabiliser fragment of quantum theory is a foundational building  block for quantum error correction and the fault-tolerant compilation  of quantum programs. In this article, we develop a sound, universal and complete  denotational semantics for stabiliser operations which include measurement,  classically-controlled Pauli operators, and affine classical operations,  in which quantum error-correcting codes are  first-class objects. The operations are interpreted as certain \emph{affine
  relations} over finite fields. This offers a conceptually motivated and computationally-tractable alternative to the standard  operator-algebraic semantics of quantum programs (whose time complexity grows exponentially as the state space increases in size).   We demonstrate the power of the resulting semantics by describing a   small, proof-of-concept assembly language for stabiliser programs with  fully-abstract denotational semantics. 
\end{abstract}
\section{Introduction}

The problem of compiling quantum algorithms into fault-tolerant hardware-level
instructions is a central challenge in the design of scalable quantum systems
\cite{campbell_roads_2017, beverland_assessing_2022, Paler2017}. Quantum
error-correcting codes are central to this challenge, among which
stabiliser codes are the most common and well-studied \cite{0904.2557}.
For fault‑tolerant compilation to scale, we need a better understanding
of the compositional structure of fault-tolerance, and therefore of the
stabiliser fragment. Unlike general quantum programs, stabiliser quantum
programs can be simulated efficiently on a probabilistic classical computer
\cite{aaronson2004improved}. Despite this fact, the formal denotational
semantics of stabiliser quantum programs has not been thoroughly studied. That is to say, the mathematical structure of quantum programs built from stabiliser operations is not well-understood. 

In this article, we develop a nondeterministic denotational semantics
for quantum programs built from stabiliser operations, including Clifford
operators, Pauli errors, Pauli measurement, affine classical operations and
classically-controlled Pauli operators. Our finely tuned semantics is to be
contrasted with the usual, much larger denotational semantics of non-stabiliser
quantum programs in terms of quantum channels \cite{mingsheng_foundations_2016}. Our work draws from two lines
of research: the categorical semantics of quantum programming languages and
quantum computing \cite{selinger_towards_2004, selinger_idempotents_2008,
selinger_dagger_2007, heunen_completely_2014}; and the symplectic
representation of pure stabiliser circuits \cite{gross_finite_2005,
neretin_lectures_2011, heinrich, comfort_graphical_2021,
booth_graphical_2024}. Ultimately, these results constitute the first
step towards the development of formally verified \emph{fault-tolerant}
quantum compilation frameworks, integrating current approaches to
compilation \cite{cross_improved_2024,cowtan_css_2024, he_extractors_2025,
poirson_engineering_2025} and verification \cite{rand_gottesman_2020,
chen_verifying_2025, huang_efficient_2025, wu_qecv_2021, sundaram_hoare_2025,
fang_symbolic_2024, qwire, reqwire}.

The categorical semantics of quantum theory builds on the mathematical
formulation of finite-dimensional quantum processes with measurement. This
formulation can be rigorously stated in the language of operator algebras
\cite{Busch2016} in three stages of increasing expressivity:
\begin{enumerate}
  \item \emph{Pure quantum mechanics} via linear maps between finite-dimensional Hilbert spaces;
  \item \emph{Mixed quantum mechanics} via completely-positive maps between matrix algebras;
  \item \emph{Quantum measurements and classical control} via completely-positive maps between finite-dimensional \(C^*\)-algebras.
\end{enumerate}
These increasing stages of expressivity can be restated  by applying
the following functorial constructions to the \dag-compact-closed category, \(\FHilb\), of finite-dimensional Hilbert
spaces and linear maps:
\begin{equation*}
  \begin{tikzcd}[column sep=1.5cm]
    {\text{\emph{pure} QM}}
    &&
    {\text{\emph{mixed} QM}}
    &&
    {\text{QM with \emph{measurements}}}
    \arrow["\text{CPM construction}", "\text{\cite{selinger_dagger_2007}}"', from=1-1, to=1-3]
    \arrow["\text{split \dag-idempotents}", "\text{\cite{selinger_idempotents_2008}}"', from=1-3, to=1-5]
  \end{tikzcd}
\end{equation*}

Finite-dimensional quantum mechanics can therefore
be understood in purely categorical terms, agnostic to the theory of operator algebras.
This point of view is highly amenable to generalisation and
specialisation: simply replace \(\FHilb\) with any other \dag-compact-closed
category, and apply these constructions to add abstract notions of mixing
and measurement.

In this article, we work with \dag-compact-closed categories specifically tailored to the stabiliser fragment. The first semantics is obtained directly by restricting \(\FHilb\) to the stabiliser fragment; whereas the second semantics is given by the symplectic representation of stabiliser maps. Specifically, we work throughout with odd-prime-dimensional quantum systems, which ensures that the symplectic representation is well-behaved \cite[Chapter 9]{neretin_lectures_2011}. Whilst the full symplectic semantics breaks down in even characteristic, we can nevertheless recover the theory of CSS codes in the qubit case \cite{Calderbank,Steane,comfortthesis, kissinger_phase-free_2022}.

\vspace{-4mm}
\subparagraph{Outline.}
Section~\ref{section:categories} recalls the basic theory of symmetric monoidal categories and \dag-compact closed examples.
Section~\ref{section:stabiliser} begins with a review of the stabiliser formalism and its symplectic formulation. We then describe novel denotational semantics for mixed stabiliser processes in section~\ref{section:mixed} and stabiliser processes with measurement in section~\ref{section:measurement}.  Throughout  sections~\ref{section:stabiliser}-\ref{section:measurement}, we develop the standard operator-theoretic semantics given by restriction, and the corresponding symplectic representation, in each section proving their equivalence, i.e.:
\[
  \renewcommand{\arraystretch}{.55}
  \begin{tikzcd}[column sep=3.66cm, row sep=.24cm]
    {\text{pure stabiliser theory}} & {\text{affine Lagrangian relations}} \\ 
    {\text{mixed stabiliser theory}} & {\text{affine coisotropic relations}} \\ 
    {\begin{array}{c} \text{stabiliser theory} \\ \text{with measurements} \end{array}} & {\begin{array}{c} \text{affine relations} \\ \text{with symplectic modality} \end{array}}
    \arrow[
    "\text{Section~\ref{section:stabiliser}}",
      "\text{\cite[Chapter~9]{neretin_lectures_2011}, \cite{comfort_graphical_2021}}"',
      leftrightarrow,
      from=1-1,
      to=1-2
    ]
    \arrow["\text{Section~\ref{section:mixed}}", leftrightarrow, from=2-1, to=2-2]
    \arrow["\text{Section~\ref{section:measurement}}", leftrightarrow, from=3-1, to=3-2]
    \arrow["\text{ CPM construction }"', from=1-1, to=2-1]
    \arrow["\text{ CPM construction }", from=1-2, to=2-2]
    \arrow["\text{ split \dag-idempotents }"', from=2-1, to=3-1]
    \arrow["\text{ split \dag-idempotents }", from=2-2, to=3-2]
  \end{tikzcd}
\]
Finally, in section~\ref{section:language}, we define a simple imperative
language for stabiliser quantum programs, including Pauli measurement, affine classical operations,  and
classically-controlled Pauli operators, equipped with a fully abstract denotational
semantics derived from section~\ref{section:measurement}.
\vspace{-4mm}
\subparagraph{Contributions.}
We present several novel contributions:
\begin{itemize}[nosep]
    \item Corollary~\ref{corollary:cpm_acr}: a symplectic, relational semantics for completely positive stabiliser maps;
    \item Theorem~\ref{theorem:splitting_modality}: we model stabiliser quantum measurements and classical control as affine relations augmented with a modality to represent quantum data;
    \item Propositions~\ref{proposition:causal_total_cpm},\ref{proposition:causal_total_split}: we prove that the physically-realisable stabiliser programs, i.e. \emph{stabiliser quantum channels}, are represented by the total relations;
    \item Theorem~\ref{theorem:full_abstraction}: we interpret a toy programming language in this relational semantics, and prove full abstraction.
\end{itemize}
We construct finely-tuned, yet equivalent, categories of relations which offer an algebraically simpler and computationaly tractable alternative to the standard operator-theoretic semantics, whilst supporting concrete computational tools native to the stabiliser formalism and stabiliser quantum error correction. For instance, observational equivalence is decidable in polynomial time in our symplectic semantics, but in the operator-theoretic semantics it is exponential.

\section{Preliminaries: symmetric monoidal and \dag-compact closed categories}
\label{section:categories}
In this section, we review the basic theory of symmetric monoidal categories and \dag-compact closed categories.  These will serve as the mathematical structures with which we model theories of \emph{processes}, and allowing us to formally state our denotational semantics. 

One should keep in mind that the category \(\FHilb\) of finite dimensional Hilbert spaces and linear maps has all of the structure which is exposed abstractly in this section; namely, linear maps can be composed in sequence, composed in parallel using the tensor product, they can be reversed using the Hermitian adjoint, and they can also be turned into states via the Choi–Jamio\l kowski isomorphism. However, there is nothing special about \(\FHilb\) in this regard; these operations can be performed in completely different mathematical settings. Therefore, we give the structural, abstract definitions below so that we can formally state the representations which we develop in this paper.

As a technical note, the categories we are working with carry with them certain coherence isomorphisms. 
However, importantly, these isomorphisms can be turned into equalities, meaning that they  can effectively be ignored for the purpose of calculation. We give the definitions with the isomorphisms omitted, sweeping the technicalities under the rug.

\subsection{Symmetric monoidal categories}

We begin by recalling the definition of a symmetric monoidal category. These describe theories of processes which can be composed in sequence and in parallel, where moreover parallel systems can be freely exchanged with each other:

\begin{definition}
  A \textbf{strict symmetric monoidal category} (SMC) is a category \(\mathcal{C}\)
  equipped with a functor
  \(
    \otimes : \mathcal{C} \times \mathcal{C} \to \mathcal{C},
  \)
  called the \textbf{monoidal product}, and an object \(I \in \mathcal{C}\), called the
  \textbf{monoidal unit}. For all objects \(X,Y,Z\), we impose that:
  \[(X \otimes Y) \otimes Z = X \otimes (Y \otimes Z), \qquad\text{and}\qquad I \otimes X = X = X \otimes I,\]
   in addition to a natural isomorphism, the \textbf{symmetry}:
  \[
    \swap_{X,Y} : X \otimes Y \to Y \otimes X\qquad\text{such that}\qquad \swap_{X,Y} = \swap_{Y,X}^{-1}.
  \]
\end{definition}

Strict  symmetric monoidal categories admit a graphical representation using \textbf{string diagrams}.
Objects are drawn as wires and morphisms as boxes. So for example a morphism \(f:A\to B\) is drawn as \begin{tikzpicture}
	\begin{pgfonlayer}{nodelayer}
		\node [style=none] (42) at (-1, 0.525) {};
		\node [style=none] (43) at (-1, -0.525) {};
		\node [style=none] (44) at (1, 0.525) {};
		\node [style=none] (45) at (1, -0.525) {};
		\node [style=wirelable] (46) at (0.75, 0.25) {${B}$};
		\node [style=wirelable] (47) at (-0.75, 0.25) {${A}$};
		\node [style=none] (48) at (-1, 0) {};
		\node [style=none] (49) at (1, 0) {};
		\node [style=box] (50) at (0, 0) {$f$};
	\end{pgfonlayer}
	\begin{pgfonlayer}{edgelayer}
		\draw [style=background] (44.center)
			 to (45.center)
			 to (43.center)
			 to (42.center)
			 to cycle;
		\draw [style=bbox] (49.center) to (48.center);
	\end{pgfonlayer}
\end{tikzpicture}.
We think of the wires as the \emph{systems}, where the boxes are \emph{processes} acting on the systems on the left to produce systems on the right.

Given any objects \(A,B\)  and \(C\); and maps \(f : A \to B\), \(g : B \to C\), and \(h : C \to D\); the composite \(g \circ f\) is drawn by horizontal horizontal pasting; whereas, the monoidal product \(f\otimes g\) is drawn by vertical stacking: 
\begin{center}
 \begin{tikzpicture}
	\begin{pgfonlayer}{nodelayer}
		\node [style=wirelable] (5) at (1.5, 0.25) {${C}$};
		\node [style=box] (8) at (0.75, 0) {$g$};
		\node [style=none] (9) at (-1.75, 0.65) {};
		\node [style=none] (10) at (-1.75, -0.65) {};
		\node [style=none] (11) at (1.75, 0.65) {};
		\node [style=none] (12) at (1.75, -0.65) {};
		\node [style=wirelable] (13) at (0, 0.25) {${B}$};
		\node [style=wirelable] (14) at (-1.5, 0.25) {${A}$};
		\node [style=none] (15) at (-1.75, 0) {};
		\node [style=none] (16) at (1.75, 0) {};
		\node [style=box] (17) at (-0.75, 0) {$f$};
	\end{pgfonlayer}
	\begin{pgfonlayer}{edgelayer}
		\draw [style=background] (11.center)
			 to (12.center)
			 to (10.center)
			 to (9.center)
			 to cycle;
		\draw [style=bbox] (16.center) to (15.center);
	\end{pgfonlayer}
\end{tikzpicture}
 \quad \quad and \quad\quad\quad
 \begin{tikzpicture}
	\begin{pgfonlayer}{nodelayer}
		\node [style=none] (9) at (-1, 1.15) {};
		\node [style=none] (10) at (-1, -1.15) {};
		\node [style=none] (11) at (1, 1.15) {};
		\node [style=none] (12) at (1, -1.15) {};
		\node [style=wirelable] (13) at (0.75, 0.75) {${B}$};
		\node [style=wirelable] (14) at (-0.75, 0.75) {${A}$};
		\node [style=none] (15) at (-1, 0.5) {};
		\node [style=none] (16) at (1, 0.5) {};
		\node [style=box] (17) at (0, 0.5) {$f$};
		\node [style=wirelable] (18) at (-0.75, -0.75) {${C}$};
		\node [style=box] (19) at (0, -0.5) {$h$};
		\node [style=none] (20) at (-1, -0.5) {};
		\node [style=none] (21) at (1, -0.5) {};
		\node [style=wirelable] (22) at (0.75, -0.75) {${D}$};
	\end{pgfonlayer}
	\begin{pgfonlayer}{edgelayer}
		\draw [style=background] (11.center)
			 to (12.center)
			 to (10.center)
			 to (9.center)
			 to cycle;
		\draw [style=bbox] (16.center) to (15.center);
		\draw [style=bbox] (21.center) to (20.center);
	\end{pgfonlayer}
\end{tikzpicture}.
\end{center}
The symmetry \(\swap_{A,B}\) is drawn by crossing wires: \begin{tikzpicture}
	\begin{pgfonlayer}{nodelayer}
		\node [style=none] (7) at (-0.375, 0.25) {};
		\node [style=none] (9) at (0.625, 0.25) {};
		\node [style=none] (10) at (-0.375, -0.25) {};
		\node [style=none] (11) at (0.625, -0.25) {};
		\node [style=none] (18) at (0.625, 0.375) {};
		\node [style=none] (19) at (0.625, -0.375) {};
		\node [style=none] (20) at (-0.375, 0.375) {};
		\node [style=none] (21) at (-0.375, -0.375) {};
		\node [style=wirelable] (24) at (-0.625, 0.25) {${A}$};
		\node [style=wirelable] (25) at (-0.625, -0.25) {${B}$};
		\node [style=wirelable] (26) at (0.875, 0.25) {${A}$};
		\node [style=wirelable] (27) at (0.875, -0.25) {${B}$};
	\end{pgfonlayer}
	\begin{pgfonlayer}{edgelayer}
		\draw [style=background] (21.center)
			 to (19.center)
			 to (18.center)
			 to (20.center)
			 to cycle;
		\draw [in=-180, out=0, looseness=1.50] (7.center) to (11.center);
		\draw [in=0, out=-180, looseness=1.50] (9.center) to (10.center);
	\end{pgfonlayer}
\end{tikzpicture}.
The coherence conditions are reflected by the following isotopies, so that for all objects \(A,B,C\), and \(D\), and maps \(f:A\to B\) and \(g:C\to D\):

\begin{center}
 \begin{tikzpicture}
	\begin{pgfonlayer}{nodelayer}
		\node [style=none] (0) at (0.5, 0.75) {};
		\node [style=none] (1) at (0.5, -0.75) {};
		\node [style=none] (2) at (3, 0.75) {};
		\node [style=none] (3) at (3, -0.75) {};
		\node [style=none] (4) at (0.5, 0.25) {};
		\node [style=none] (5) at (1.75, 0.25) {};
		\node [style=none] (6) at (0.5, -0.25) {};
		\node [style=none] (7) at (1.75, -0.25) {};
		\node [style=none] (8) at (3, -0.25) {};
		\node [style=none] (9) at (3, 0.25) {};
		\node [style=none] (10) at (3.75, 0) {$=$};
		\node [style=none] (11) at (4.5, 0.25) {};
		\node [style=none] (12) at (6, 0.25) {};
		\node [style=none] (13) at (4.5, -0.25) {};
		\node [style=none] (14) at (6, -0.25) {};
		\node [style=none] (15) at (11.25, 0.125) {$=$};
		\node [style=box] (16) at (8.25, 0.375) {$f$};
		\node [style=none] (17) at (8.875, 0.375) {};
		\node [style=none] (18) at (7.25, 0.375) {};
		\node [style=none] (19) at (9.875, -0.375) {};
		\node [style=none] (20) at (10.5, -0.375) {};
		\node [style=none] (21) at (9.875, 0.375) {};
		\node [style=none] (22) at (10.5, 0.375) {};
		\node [style=none] (23) at (7.25, -0.375) {};
		\node [style=none] (24) at (8.875, -0.375) {};
		\node [style=none] (25) at (4.5, 0.75) {};
		\node [style=none] (26) at (4.5, -0.75) {};
		\node [style=none] (27) at (6, 0.75) {};
		\node [style=none] (28) at (6, -0.75) {};
		\node [style=none] (29) at (7.25, 1) {};
		\node [style=none] (30) at (7.25, -1) {};
		\node [style=none] (31) at (10.5, 1) {};
		\node [style=none] (32) at (10.5, -1) {};
		\node [style=box] (33) at (14.25, -0.375) {$f$};
		\node [style=none] (34) at (13.625, -0.375) {};
		\node [style=none] (35) at (15.25, -0.375) {};
		\node [style=none] (36) at (12.625, 0.375) {};
		\node [style=none] (37) at (12, 0.375) {};
		\node [style=none] (38) at (12.625, -0.375) {};
		\node [style=none] (39) at (12, -0.375) {};
		\node [style=none] (40) at (15.25, 0.375) {};
		\node [style=none] (41) at (13.625, 0.375) {};
		\node [style=none] (42) at (15.25, -1) {};
		\node [style=none] (43) at (15.25, 1) {};
		\node [style=none] (44) at (12, -1) {};
		\node [style=none] (45) at (12, 1) {};
		\node [style=wirelable] (46) at (7.5, 0.625) {${A}$};
		\node [style=wirelable] (47) at (10.25, -0.625) {${B}$};
		\node [style=wirelable] (48) at (7.5, -0.625) {${C}$};
		\node [style=wirelable] (49) at (10.25, 0.625) {${C}$};
		\node [style=wirelable] (50) at (12.25, 0.625) {${A}$};
		\node [style=wirelable] (51) at (15, -0.625) {${B}$};
		\node [style=wirelable] (52) at (12.25, -0.625) {${C}$};
		\node [style=wirelable] (53) at (15, 0.625) {${C}$};
		\node [style=wirelable] (65) at (0.75, 0.5) {${A}$};
		\node [style=wirelable] (66) at (0.75, -0.5) {${B}$};
		\node [style=wirelable] (67) at (2.75, 0.5) {${A}$};
		\node [style=wirelable] (68) at (2.75, -0.5) {${B}$};
		\node [style=wirelable] (69) at (5.25, 0.5) {${A}$};
		\node [style=wirelable] (70) at (5.25, -0.5) {${B}$};
		\node [style=wirelable] (71) at (1.75, -0.5) {${A}$};
		\node [style=wirelable] (72) at (1.75, 0.5) {${B}$};
		\node [style=none] (73) at (6.5, 0) {,};
		\node [style=none] (74) at (16.25, 0) {and};
		\node [style=box] (75) at (18.25, 0.375) {$f$};
		\node [style=none] (76) at (17.25, -0.375) {};
		\node [style=none] (77) at (20.25, -0.375) {};
		\node [style=none] (79) at (17.25, 0.375) {};
		\node [style=none] (82) at (20.25, 0.375) {};
		\node [style=none] (84) at (20.25, -1) {};
		\node [style=none] (85) at (20.25, 1) {};
		\node [style=none] (86) at (17.25, -1) {};
		\node [style=none] (87) at (17.25, 1) {};
		\node [style=wirelable] (88) at (17.5, 0.625) {$A$};
		\node [style=wirelable] (89) at (20, -0.625) {$D$};
		\node [style=wirelable] (90) at (17.5, -0.625) {$C$};
		\node [style=wirelable] (91) at (20, 0.625) {$B$};
		\node [style=box] (92) at (19.25, -0.375) {$g$};
		\node [style=none] (93) at (21, 0.125) {$=$};
		\node [style=box] (94) at (23.75, 0.375) {$f$};
		\node [style=none] (95) at (21.75, -0.375) {};
		\node [style=none] (96) at (24.75, -0.375) {};
		\node [style=none] (97) at (21.75, 0.375) {};
		\node [style=none] (98) at (24.75, 0.375) {};
		\node [style=none] (99) at (24.75, -1) {};
		\node [style=none] (100) at (24.75, 1) {};
		\node [style=none] (101) at (21.75, -1) {};
		\node [style=none] (102) at (21.75, 1) {};
		\node [style=wirelable] (103) at (22, 0.625) {$A$};
		\node [style=wirelable] (104) at (24.5, -0.625) {$D$};
		\node [style=wirelable] (105) at (22, -0.625) {$C$};
		\node [style=wirelable] (106) at (24.5, 0.625) {$B$};
		\node [style=box] (107) at (22.75, -0.375) {$g$};
	\end{pgfonlayer}
	\begin{pgfonlayer}{edgelayer}
		\draw [style=background] (30.center)
			 to [in=270, out=90] (29.center)
			 to [in=180, out=0] (31.center)
			 to [in=90, out=-90] (32.center)
			 to [in=360, out=180] cycle;
		\draw [style=background] (26.center)
			 to (25.center)
			 to (27.center)
			 to (28.center)
			 to cycle;
		\draw [style=background] (3.center)
			 to (1.center)
			 to (0.center)
			 to (2.center)
			 to cycle;
		\draw (4.center)
			 to [in=-180, out=0, looseness=1.50] (7.center)
			 to [in=180, out=0, looseness=1.50] (9.center);
		\draw (8.center)
			 to [in=360, out=180, looseness=1.50] (5.center)
			 to [in=0, out=-180, looseness=1.50] (6.center);
		\draw (11.center) to (12.center);
		\draw (13.center) to (14.center);
		\draw (20.center)
			 to (19.center)
			 to [in=0, out=-180, looseness=1.25] (17.center)
			 to (18.center);
		\draw (22.center)
			 to (21.center)
			 to [in=0, out=-180, looseness=1.25] (24.center)
			 to (23.center);
		\draw [style=white] (3.center) to (1.center);
		\draw [style=background] (28.center) to (26.center);
		\draw [style=background] (32.center) to (30.center);
		\draw [style=background] (43.center)
			 to [in=90, out=-90] (42.center)
			 to (44.center)
			 to [in=-90, out=90] (45.center)
			 to [in=180, out=0] cycle;
		\draw (35.center)
			 to (34.center)
			 to [in=0, out=180, looseness=1.25] (36.center)
			 to (37.center);
		\draw (40.center)
			 to (41.center)
			 to [in=0, out=180, looseness=1.25] (38.center)
			 to (39.center);
		\draw [style=background] (45.center) to (43.center);
		\draw [style=background] (85.center)
			 to [in=90, out=-90] (84.center)
			 to (86.center)
			 to [in=-90, out=90] (87.center)
			 to cycle;
		\draw [style=background, in=90, out=-90] (85.center) to (84.center);
		\draw [style=background] (84.center) to (86.center);
		\draw [style=background, in=-90, out=90] (86.center) to (87.center);
		\draw [style=background, in=180, out=0] (87.center) to (85.center);
		\draw (77.center) to (76.center);
		\draw [style=background] (87.center) to (85.center);
		\draw (79.center) to (82.center);
		\draw [style=background] (100.center)
			 to [in=90, out=-90] (99.center)
			 to (101.center)
			 to [in=-90, out=90] (102.center)
			 to cycle;
		\draw [style=background, in=90, out=-90] (100.center) to (99.center);
		\draw [style=background] (99.center) to (101.center);
		\draw [style=background, in=-90, out=90] (101.center) to (102.center);
		\draw [style=background, in=180, out=0] (102.center) to (100.center);
		\draw (96.center) to (95.center);
		\draw [style=background] (102.center) to (100.center);
		\draw (97.center) to (98.center);
	\end{pgfonlayer}
\end{tikzpicture}.
 \end{center}

\begin{example}
    Given a semiring \(R\), the category \(\sf{Mat}_R\) has objects given by the natural numbers and morphisms given by matrices over \(R\). It is a strict SMC where the monoidal unit is \(0\), the monoidal product is given by the Kronecker product, and the symmetries are given by permutation matrices.
\end{example}

Non-strict  SMCs are defined by replacing the equalities above with coherent natural isomorphisms; where symmetric monoidal functors preserve the relevant structure coherently.  A symmetric monoidal functor is a functor which coherently preserves this structure.
Because every SMC is symmetric monoidally equivalent to a strict one, we abuse notation, and ignore these isomorphisms throughout the paper.

The following example is the first approximation of a quantum semantics which we work with:

\begin{example}
The symmetric monoidal category, \(\FHilb\), of \textbf{finite-dimensional Hilbert spaces and linear maps} has:
\begin{itemize}
  \item \textbf{objects:} finite-dimensional Hilbert spaces \(\mathcal H\), i.e.\ finite-dimensional complex vector spaces equipped with an inner product \(\langle - | = \rangle\) antilinear in the first argument and linear in the second argument;
  \item \textbf{morphisms, identities and composition} given by linear maps, the identity linear map and the composition of linear maps;
  \item \textbf{monoidal product:} the bilinear tensor product \(\mathcal H \otimes \mathcal K\);
  \item \textbf{monoidal unit:} \(\C\);
  \item \textbf{symmetry:} for basis vectors \(\ket{h} \in \mathcal H\), \(\ket{k} \in \mathcal K\),
  \(
     \swap_{\mathcal H,\mathcal K}(\ket{h} \otimes \ket{k})
     \coloneqq
     \ket{k} \otimes \ket{h}
  \).
  
  Note that this not depend on the chosen basis.
\end{itemize}
\end{example}
Whereas, the next example serves as a semantics for nondeterministic processes:
\begin{example}
The symmetric monoidal category, \(\FinRel\), of \textbf{finite sets and relations}:
\begin{itemize}
  \item \textbf{objects:} finite sets;
  \item \textbf{morphisms:} relations \(R \subseteq A \times B\);
  \item \textbf{identities:} diagonal relations \(1_A \coloneqq \{(a,a) \mid a \in A\}\);
  \item \textbf{composition:} for \(R \subseteq A \times B\) and \(S \subseteq B \times C\), the relational composite is given by:
  \[
     S \circ R \coloneqq \{(a,c) \mid \exists b \in B \text{ with } (a,b)\in R \text{ and } (b,c)\in S\} \subseteq A\times C;
  \]
  \item \textbf{monoidal product:} Cartesian product \(A \times B\);
  \item \textbf{monoidal unit:} the singleton set \(\{\bullet\}\);
  \item \textbf{symmetry:}
  \(
     \swap_{A,B} \coloneqq \{((a,b),(b,a)) \mid a \in A, b \in B\} .
  \)
\end{itemize}
\end{example}

Given a symmetric monoidal category \(\mathsf{C}\), the homset
\(\mathsf{C}(I,I)\) forms a monoid called the \textbf{scalars} of
\(\mathsf{C}\). Let \(\Proj(\mathsf{C})\) denote the symmetric monoidal
category obtained by quotienting by the subgroup of invertible scalars.
\begin{example}
  In \(\FHilb\), the group of invertible scalars is the multiplicative subgroup
  \(\C^\times\subseteq \C\), so that \(\Proj(\FHilb)\) has finite-dimensional Hilbert spaces as objects and \emph{projectivised linear maps} as morphisms,
  i.e. equivalence classes of linear maps with respect to the relation
  \begin{equation}
    \phi \sim \psi \qq{if and only if} \exists \lambda \in \C^\times:\ \phi = \lambda \cdot \psi.
  \end{equation}
  In \(\FinRel\) the monoid of scalars is the trivial group, so
  \(\Proj(\FinRel) = \FinRel\).
\end{example}

To the reader who is uncomfortable with the fact that we have ignored the coherence isomorphisms, both examples can be made into the following concrete strict SMCs:
\begin{lemma}
\(\FHilb\) is symmetric monoidally equivalent to matrices over the complex numbers.
\end{lemma}

\begin{lemma}
\(\FinRel\) is symmetric monoidally equivalent to matrices over the booleans, 
\end{lemma}

\subsection{\dag-symmetric monoidal categories}

In the previous subsection, we reviewed the theory of symmetric monoidal categories: which model processes which can be composed in sequence, in parallel, and where systems can be freely exchanged with each other. In this section, we add additional structure which allows for processes to be reversed:

\begin{definition}
  A strict \textbf{\dag-symmetric monoidal category} (\dag-SMC) is a symmetric
  monoidal category equipped with a strict SMC functor \( (-)^\dag : \mathsf C^{\op}
  \to \mathsf C \) that is identity on objects, and for any map \(f\), \( (f^\dag)^\dag =
  f \).
\end{definition}

\begin{definition}
  A map \(f:A\to B\) in \(\mathsf{C}\) is:
  \begin{itemize}
      \item an \textbf{isometry} if \(f^\dag \circ f = 1_A\);
      \item a \textbf{coisometry} if \(f \circ f^\dag = 1_B\);
      \item a \textbf{unitary} if \(f^\dag = f^{-1}\).
  \end{itemize}
\end{definition}

Again the nonstrict version, is defined by replacing some of the equalities by coherent natural isomorphism, and the notion of symmetric monoidal functor preserves this structure coherently.
Both of our running examples are \dag-symmetric monoidal categories:

\begin{example}
 \(\FHilb\) is a \dag-symmetric monoidal  category with respect to the Hermitian adjoint.
Concretely give a linear map \(f : \mathcal H \to \mathcal K\), the Hermitian adjoint \(f^\dag : \mathcal K \to \mathcal H\) is the unique map such that for all \(x \in \mathcal H\) and \(y \in \mathcal K\),
\(
\langle f x | y \rangle = \langle x | f^\dag y \rangle
\).
The isometries, coisometries and unitaries in the \dag-symmetric monoidal sense, correspond with the usual notions in linear algebra.
\end{example}

\begin{example}
 \(\FinRel\) is a \dag-symmetric monoidal category with respect to the relational converse,
\(
R^\dag \coloneqq \{(b,a) \mid (a,b) \in R\}
\).
The unitaries are isomorphisms.
\end{example}




\subsection{\dag-compact closed categories}

In the previous subsection, we reviewed the theory of \dag-symmetric monoidal categories, modelling processes which can be composed in parallel, in sequence, where systems can be freely exchanged with each other, and where processes can be reversed. In this section, we ask for additional structure which accommodates a notion of feedback, so that the outputs of processes can be turned into inputs and vice-versa:

\begin{definition}
A strict \dag-symmetric monoidal category is \textbf{\dag-compact closed} (\dag-CCC) when every object \(A\) has a chosen \textbf{dual} object \(A^*\), extending to a symmetric monoidal functor \((-)^*:\mathsf{C}^\op\to \mathsf{C}\);
where moreover, we impose that for all objects \(A\), there are morphisms \(
\eta_A : I \to A \otimes A^*\) and \(
\epsilon_A : A^* \otimes A \to I\) satisfying
\begin{itemize}
  \item the \textbf{snake equations},
  \[
  (1_{A^*} \otimes \epsilon_A)
  \circ
  (\eta_A \otimes 1_{A^*})
  =
  1_{A^*},
\qquad\text{and}\qquad
  (\epsilon_A \otimes 1_{A})
  \circ
  (1_{A} \otimes \eta_A)
  =
  1_{A};\]

  \item and \textbf{compatibility with the dagger},
  \(
  \epsilon_A^\dag
  =
  \swap_{A,A^*}
  \circ
  \eta_A
  \).
\end{itemize}

In string diagrams, that is when we can bend wires as follows:


\begin{center}
    \begin{tikzpicture}
	\begin{pgfonlayer}{nodelayer}
		\node [style=none] (0) at (20.75, 0.875) {};
		\node [style=none] (1) at (20.75, -0.875) {};
		\node [style=none] (2) at (22, 0.875) {};
		\node [style=none] (3) at (22, -0.875) {};
		\node [style=none] (8) at (9.75, 1) {};
		\node [style=none] (9) at (9.75, -1) {};
		\node [style=none] (10) at (12.75, 1) {};
		\node [style=none] (11) at (12.75, -1) {};
		\node [style=none] (12) at (6.5, 0.5) {};
		\node [style=none] (13) at (6.5, -0.5) {};
		\node [style=none] (14) at (8.25, 0.5) {};
		\node [style=none] (15) at (8.25, -0.5) {};
		\node [style=none] (16) at (2, 1) {};
		\node [style=none] (17) at (2, -1) {};
		\node [style=none] (18) at (5, 1) {};
		\node [style=none] (19) at (5, -1) {};
		\node [style=none] (20) at (6.5, 0) {};
		\node [style=none] (21) at (8.25, 0) {};
		\node [style=wirelable] (22) at (7.375, 0.225) {${A}$};
		\node [style=none] (23) at (5.75, 0) {$=$};
		\node [style=none] (24) at (2.5, 0) {};
		\node [style=none] (25) at (3.25, -0.75) {};
		\node [style=none] (26) at (3.75, -0.75) {};
		\node [style=none] (27) at (3.75, 0) {};
		\node [style=none] (28) at (3.25, 0) {};
		\node [style=none] (29) at (3.25, 0.75) {};
		\node [style=none] (30) at (3.75, 0.75) {};
		\node [style=none] (31) at (4.5, 0) {};
		\node [style=none] (32) at (2, 0) {};
		\node [style=none] (33) at (5, 0) {};
		\node [style=none] (34) at (12.25, 0) {};
		\node [style=none] (35) at (11.5, -0.75) {};
		\node [style=none] (36) at (11, -0.75) {};
		\node [style=none] (37) at (11, 0) {};
		\node [style=none] (38) at (11.5, 0) {};
		\node [style=none] (39) at (11.5, 0.75) {};
		\node [style=none] (40) at (11, 0.75) {};
		\node [style=none] (41) at (10.25, 0) {};
		\node [style=none] (42) at (12.75, 0) {};
		\node [style=none] (43) at (9.75, 0) {};
		\node [style=wirelable] (44) at (4.75, 0.225) {${A}$};
		\node [style=wirelable] (45) at (2.25, 0.225) {${A}$};
		\node [style=none] (46) at (14.25, 0.5) {};
		\node [style=none] (47) at (14.25, -0.5) {};
		\node [style=none] (48) at (16, 0.5) {};
		\node [style=none] (49) at (16, -0.5) {};
		\node [style=none] (50) at (14.25, 0) {};
		\node [style=none] (51) at (16, 0) {};
		\node [style=wirelable] (52) at (15.125, 0.225) {${A^*}$};
		\node [style=none] (53) at (13.5, 0) {$=$};
		\node [style=none] (54) at (18.25, 0) {{and}};
		\node [style=none] (55) at (8.75, -0.125) {,};
		\node [style=none] (64) at (20.75, -0.375) {};
		\node [style=none] (65) at (21.25, -0.375) {};
		\node [style=none] (66) at (21.25, 0.375) {};
		\node [style=none] (67) at (20.75, 0.375) {};
		\node [style=wirelable] (68) at (21, 0.65) {${A}^*$};
		\node [style=wirelable] (69) at (21, -0.65) {${A}$};
		\node [style=none] (70) at (20.75, -0.375) {};
		\node [style=none] (71) at (20.75, 0.375) {};
		\node [style=none] (73) at (25.5, -0.125) {.};
		\node [style=none] (78) at (25, 0.875) {};
		\node [style=none] (79) at (25, -0.875) {};
		\node [style=none] (80) at (23.25, 0.875) {};
		\node [style=none] (81) at (23.25, -0.875) {};
		\node [style=none] (82) at (24, -0.375) {};
		\node [style=none] (83) at (24, 0.375) {};
		\node [style=wirelable] (84) at (24, -0.65) {${A}^*$};
		\node [style=wirelable] (85) at (24, 0.65) {${A}$};
		\node [style=none] (86) at (25, 0.375) {};
		\node [style=none] (87) at (24, -0.375) {};
		\node [style=none] (88) at (24, 0.375) {};
		\node [style=none] (89) at (25, -0.375) {};
		\node [style=none] (98) at (22.5, 0) {$=$};
		\node [style=none] (99) at (16.5, -0.125) {,};
		\node [style=wirelable] (100) at (12.5, 0.225) {${A}^*$};
		\node [style=wirelable] (101) at (10, 0.225) {${A}^*$};
		\node [style=none] (102) at (22.25, 0.875) {\(\dag\)};
	\end{pgfonlayer}
	\begin{pgfonlayer}{edgelayer}
		\draw [style=background] (0.center)
			 to (2.center)
			 to (3.center)
			 to (1.center)
			 to cycle;
		\draw [style=background] (8.center)
			 to (10.center)
			 to [in=90, out=-90] (11.center)
			 to (9.center)
			 to cycle;
		\draw [style=background] (12.center)
			 to (14.center)
			 to (15.center)
			 to (13.center)
			 to cycle;
		\draw [style=background] (19.center)
			 to (17.center)
			 to (16.center)
			 to (18.center)
			 to cycle;
		\draw [style=background] (47.center)
			 to (46.center)
			 to (48.center)
			 to (49.center)
			 to cycle;
		\draw (21.center) to (20.center);
		\draw (32.center)
			 to (24.center)
			 to [in=180, out=0] (25.center)
			 to [in=180, out=0] (26.center)
			 to [bend right=90, looseness=2.00] (27.center)
			 to [in=0, out=180] (28.center)
			 to [bend right=270, looseness=2.00] (29.center)
			 to [in=180, out=0] (30.center)
			 to [in=-180, out=0] (31.center)
			 to (33.center);
		\draw (43.center)
			 to (41.center)
			 to [in=180, out=0] (40.center)
			 to [in=180, out=360] (39.center)
			 to [bend left=90, looseness=1.75] (38.center)
			 to [in=360, out=180] (37.center)
			 to [bend right=90, looseness=1.75] (36.center)
			 to [in=180, out=360] (35.center)
			 to [in=180, out=0] (34.center)
			 to (42.center);
		\draw (51.center) to (50.center);
		\draw (67.center)
			 to [in=180, out=360] (66.center)
			 to [bend left=90, looseness=1.75] (65.center)
			 to [in=360, out=180] (64.center);
		\draw [style=background] (79.center)
			 to (78.center)
			 to (80.center)
			 to (81.center)
			 to cycle;
		\draw [style=background] (78.center) to (80.center);
		\draw [style=background] (80.center) to (81.center);
		\draw [style=background] (81.center) to (79.center);
		\draw [style=background] (79.center) to (78.center);
		\draw (89.center)
			 to [in=0, out=180, looseness=1.25] (88.center)
			 to (83.center)
			 to [bend right=90, looseness=1.75] (82.center)
			 to (87.center)
			 to [in=180, out=0, looseness=1.25] (86.center);
	\end{pgfonlayer}
\end{tikzpicture}
\end{center}
\end{definition}

Note that a \dag-symmetric monoidal functor preserves \dag-compact closed structure.

\begin{example}
In \(\FHilb\), the dual is given on objects \(\mathcal{H}\) by the internal hom \(\mathcal{H}^* \coloneqq [\mathcal{H}, \C]\) of linear functions from \(\mathcal{H}\) into \(\C\).
Given an orthonormal basis \(\{\ket{j}\}_{j=0}^{\dim(\mathcal{H})-1}\) of \(\mathcal H\), the inner product induces an orthonormal basis \(\{\bra{j}\}_{j=0}^{\dim(\mathcal{H})-1}\) of \(\mathcal H^*\). The cup and cap  are given by:
\[
\eta_{\mathcal H}(1) \coloneqq \sum_{j =0}^{\dim(\mathcal{H})-1} \!\!\! \ket{j}\otimes \bra{j}: \C\to  \mathcal{H} \otimes \mathcal{H}^*
\qand
\epsilon_{\mathcal{H}}(\bra{j}\otimes \ket{k}) \coloneqq \bra{j}\ket{k}:\mathcal{H}^* \otimes \mathcal{H} \to \C .
\]
Note this doesn't depend on the choice of basis. By interpreting \(\mathcal{H}\) as the space of pure quantum states, the normalised cup \(\eta_\mathcal{H}/\sqrt{\dim(\mathcal{H})}\) is interpreted as the \emph{maximally entangled state}. Therefore, the structure of \dag-compact closure is the abstract property needed to possess a notion of the Choi–Jamio\l kowski isomorphism.

In mixed state quantum mechanics, the set \(\mathcal{B}(\mathcal{H})\) of bounded linear maps from \(\mathcal{H}\) to \(\mathcal{H}\) can be regarded as a \(C^*\)-algebra, called a \emph{matrix algebra}, which in  finite dimensions is represented by the internal hom \(\mathcal{B}(\mathcal{H})\cong \mathcal{H}\otimes \mathcal{H}^*\). In this setting, the normalised cup \(\eta_\mathcal{H}/\sqrt{\dim(\mathcal{H})}\) is interpreted as the \emph{maximally mixed state} \(\rho_{\mathsf{max}}\).  On the other hand, composing the cap with the symmetry, we have the \emph{quantum trace}:
\[\Tr_{\mathcal{H}}(\ket{j}\otimes \bra{k}) \coloneqq \bra{k}\ket{j}:\mathcal{H} \otimes \mathcal{H}^* \to \C \]
In particular, given a linear map \(T:\mathcal{H}\to \mathcal{H}\) regarded as an element \(\lfloor T \rfloor: \C\to  \mathcal{H}\otimes \mathcal{H}^*\), the quantum trace coincides with the usual linear-algebraic trace \(\Tr_{\mathcal{H}}\circ \lfloor T \rfloor = \Tr(T)\).

\end{example}
There is an analogous correspondence for relations:
\begin{example}
In \(\FinRel\), the dual is trivial \(A^* = A\).
The cup and cap are given by the diagonal relations:
\[
\eta_A \coloneqq \{(\bullet,(a,a)) \mid a \in A\}: \{\bullet\}\to  A\times A,
\ \ \text{and} \ \ 
\epsilon_A \coloneqq \{((a,a),\bullet) \mid a \in A\}:  A\times A\to \{\bullet\}.
\]
\end{example}



\section{The pure stabiliser theory}
\label{section:stabiliser}
We review the key elements of the stabiliser theory, and its representation
with symplectic linear algebra. This leads to the definition of two \dag-CCCs
for the pure stabiliser theory; one in terms of finite dimensional Hilbert
spaces and the other in terms of structure-preserving relations between sets
carrying symplectic structure:
\begin{enumerate}
    \item a \emph{concrete} \dag-CCC, \(\Stab_p\), given by restricting
 \(\FHilb\);
    \item an \emph{abstract} \dag-CCC, \(\ALR\),
described in terms of affine subspaces over finite fields.
\end{enumerate}
These two \dag-CCCs are
known to be equivalent up to nonzero scalars \cite{neretin_lectures_2011,comfort_graphical_2021}. We will take \(\ALR\)
as the basis from which we build our abstract denotational semantics.

\subsection{The Hilbert space picture}
\label{section:stabiliser:hilbert}


\subparagraph{Notation.} 
Throughout this paper, let \(p\) denote an \emph{odd} prime, so that \(\F_p \coloneqq
\Z/p\Z\) is the field of integers modulo \(p\). 

Let \(\mathcal{H}_p \coloneqq
\C^p\) be the \(p\)-dimensional Hilbert space, equipped with the canonical orthonormal basis \(\{\ket{x} \mid x
\in \F_p\}\). The Hilbert space \(\mathcal{H}_p\) models the pure states of a quantum system
called a \textbf{qupit}.  

\begin{definition}
    Let \(\xi(x) \coloneqq \exp(i 2\pi x /p)\), then
    the \textbf{Pauli operators} on \(\mathcal{H}_p\) are generated by \(Z \ket{x}
    \coloneqq \xi(x) \ket{x}\) and \(X \ket{x} \coloneqq \ket{x+1}\) and assemble
    into the qupit \textbf{Pauli group}: \[\mathcal{P}_p \coloneqq \{ \xi(y) X^x
    Z^z \mid x,y,z \in \F_p\} \subseteq \mathcal{U}(\mathcal{H}_p).\] The \textbf{\(n\)-qupit Pauli group} is defined to be the
    \(n\)-fold tensor product \(\mathcal{P}_p^{\otimes n}\subseteq \mathcal{U}(\mathcal{H}_p^{\otimes n})\), so that an arbitrary
    Pauli operator takes the form \(\xi(y)  \bigotimes_{j=1}^n X_j^{x_j} Z_j^{z_j}\)
    for some \(\mathbf{x}, \mathbf{z} \in \F_p^n\) and \(y \in \F_p\). 
\end{definition}

The \emph{key observation} of the stabiliser theory is that (certain) subgroups
of \(\Pauli\) determine pure states in \(\mathcal{H}_p^{\otimes n}\):
\begin{lemma}
   Take a maximal Abelian subgroup \(S \subseteq \mathcal{P}_p^{\otimes n}\)
   such that \(\xi(a) 1_{\mathcal{H}_p}^{\otimes n} \in S\) if and only if
   \(a \equiv 0\mod p\). Up to global phase \(\exp(2\pi i \theta)\),  \(S\) uniquely determines a normalised quantum
   state \(\ket{S} \in \mathcal{H}_p^{\otimes n}\) such that \(s \ket{S} =
   \ket{S}\) for all \(s \in S\).
\end{lemma}
The equivalence class \([\exp(2\pi i \theta) \ket{S}]_{\theta\in [0,1)}\)
is called the \textbf{stabiliser state} associated to the \textbf{stabiliser
group} \(S\). The unitary operations that preserve this structure can also
be captured by their action on the Pauli group:
\begin{definition}
    The \textbf{Clifford group} is the unitary normaliser
    of \(\mathcal{P}_p^{\otimes n}\): \[\Cliff \coloneqq \{ U
    \in \mathcal{U}(\mathcal{H}_p^{\otimes n}) \mid \forall P \in
    \mathcal{P}_p^{\otimes n}, UPU^\dagger \in \mathcal{P}_p^{\otimes n}\}\subseteq \mathcal{U}(\mathcal{H}_p^{\otimes n}).\]
\end{definition}
Consider a stabiliser group \(S \subseteq \Pauli\), then for any \(C
\in \Cliff\) and \(s \in S\), we have that \(CsC^\dagger \cdot C\ket{S} =
Cs\ket{S} = C\ket{S}\). It follows that the stabiliser group \(CSC^\dagger =
\{ CsC^\dagger \mid s \in S\}\) stabilises the state \(\ket{CSC^\dagger} =
C \ket{S}\). Clifford unitaries therefore map stabiliser states to stabiliser
states. Therefore, it is natural to assemble these operations together into a \dag-CCC:

%
%
\begin{definition}
    The \dag-CCC, \(\Stab_p\), of (pure) \textbf{qupit stabiliser maps} is the \dag-compact-closed subcategory of \(\FHilb\) generated by the qupit stabiliser states and Clifford operators as well as the scalars \(1/\sqrt p\) and \(\sqrt{p}\) under tensor product, composition and the Hermitian adjoint.
\end{definition}

We add the scalars \(1/\sqrt p\) and \(\sqrt{p}\) to obtain a compact closed
category as the compact closed structure is not normalised.

\subsection{The symplectic picture}
\label{section:stabiliser:symplectic}

We recall how the theory of pure stabiliser maps can be restated in purely
symplectic terms by taking the notion of a stabiliser group, and their
symplectic representation as fundamental.


\begin{definition}
  A \textbf{symplectic vector space} \((V, \omega)\) is a (finite-dimensional)
  \(\F_p\)-vector space \(V\) equipped with a non-degenerate bilinear form
  \(\omega : V \oplus V \to \F_p\) such that for any \(v \in V\), \(\omega(v,v)
  = 0\),
\end{definition}
\begin{example}[Standard symplectic form]
  Given any  \(n \in \N\), \((\F_p^{2n},\omega_n)\)
  is a symplectic vector space where
  \(
    \omega_n\bigl((\vb{x},\vb{z}),(\vb{x'},\vb{z'})\bigr) \coloneqq \vb{x}^\trans \vb{z}' - \vb{z}^\trans \vb{x}'
  \).
\end{example}

We endow the set \(\F_p\oplus \F_p^{2n}\) with the structure of a non-Abelian group, with identity and inverse given pointwise by \(0\) and negation,  where multiplication is given by
\[(a,\vb{v})\cdot (b,\vb{w}) = (a+b,\vb{v}+ \vb{w}+\tfrac{1}{2}\omega_n(\vb{v},\vb{w})).\]
Moreover, there is a function 
\begin{equation}
  \label{eq:symplectic_presentation}
  \pi : \F_p \oplus  \F_p^{2n} \longrightarrow \Pauli ; \quad 
  \pi(a, \vb{x}, \vb{z}) \coloneqq \xi(a+\tfrac{1}{2}\vb{z}^\trans \vb{x})  \bigotimes_{k=1}^n X^{x_k} Z^{z_k},
\end{equation}
chosen such that \(\pi(a, \vb{x},\vb{z}) \pi(a', \vb{x'},\vb{z'}) =
\pi(a+a' + \tfrac{1}{2} \omega_n((\vb{x},\vb{z}),(\vb{x'}, \vb{z'})),
\vb{x}+\vb{x'},\vb{z}+\vb{z'})\), making \(\pi\) into an isomorphism of groups. Using this isomorphism,  the commutation of Paulis can immediately be framed in terms of the symplectic form:
\begin{lemma}
  \label{lem:pauli_commutation}
   Paulis \(\pi(a, \vb{x},\vb{z})\) and \(\pi(a', \vb{x'},\vb{z'})\) commute
  if and only if \(\omega_n((\vb{x},\vb{z}),
  (\vb{x'},\vb{z'})) =0\).
\end{lemma}
This allows us to represent stabiliser states in terms of maximal Abelian subgroups of \(\F_p\oplus \F_p^{2n}\). These subgroups have the convenient property that they correspond to a particular kind of affine subspaces of the symplectic vector space \((\F_p^{2n},\omega_n)\):
\begin{definition}
  Given a linear subspace \(S\) of a symplectic vector space \((V, \omega)\),
  its \textbf{symplectic complement} is the space \(S^\omega \coloneqq \{\vb{v} \in V \mid
  \forall \vb{s} \in S,\, \omega(\vb{v},\vb{s}) = 0\}\) of vectors in \(V\) which are orthogonal to those in \(S\) with respect to the symplectic form \(\omega\). A subspace \(S\subseteq (V, \omega)\) is:
  \begin{itemize}
    \item \textbf{isotropic} if \(S \subseteq S^\omega\);
    \item \textbf{coisotropic} if \(S^\omega \subseteq S\);
    \item \textbf{Lagrangian} if \(S = S^\omega\).
  \end{itemize}
  An affine subspace is isotropic, coisotropic, or Lagrangian if its linear
  component\footnote{The \textbf{linear component} of a non-empty affine subspace \(S \subseteq \F_p^n\)
  is the linear subspace\\ \(\{\vb{x}-\vb{y} \mid \vb{x},\vb{y}
  \in S\}\subseteq \F_p^n\).} is. By convention, empty subspaces are affine Lagrangian.
\end{definition}
\begin{proposition}[\cite{gross_finite_2005}]
  \label{proposition:lagrangian_to_stabiliser}
  The isomorphism \(\pi\) induces a bijection between affine Lagrangian
  subspaces of \((\F_p^{2n},\omega_n)\) and stabiliser subgroups of
  \(\Pauli\). Given a nonempty affine Lagrangian subspace \(L+\vb{a}
  \subseteq (\Zp^{2n},\omega_n)\), the associated stabiliser group is the
  Abelian subgroup \(\{\pi\bigl(\omega_n(\vb{a},\vb b),\vb{b}) \,\mid \,
  \vb{b} \in L\} \subseteq \Pauli\).
\end{proposition}
The isotropic condition on a subspace is equivalent to the commutation of the
corresponding Pauli operators. The additional Lagrangian condition imposes
that these subgroups must be maximal. The affine component accounts for the
phase factor \(\chi(a)\). Therefore,  stabiliser states can be represented using only symplectic linear algebra.


Next we give a symplectic representation of the Clifford group, starting with the following definition:
\begin{definition}
  A \textbf{symplectomorphism} \(\phi:(V,\omega_V)\xrightarrow{\cong}(W,\omega_W)\)
  is a linear isomorphism such that for all \(\vb{v},\vb{v'} \in V\),
  \(\omega_W(\phi(\vb{v}),\phi(\vb{v'})) = \omega_V(\vb{v},\vb{v'})\).
  An \textbf{affine symplectomorphism} is an affine transformation \(V \to W\) whose
  linear component\footnote{A function \(\phi\) between vector spaces is an \textbf{affine transformation} if there
  is a linear transformation \(T\) called the \textbf{linear component} of \(\phi\) such that \(\phi(\vb{x}) - \phi(\vb{0}) =
  T(x)\).\\  Given an affine transformation \(\phi\), \(\phi(\vb{0})\) is the \textbf{affine shift} of \(\phi\).}
  is a symplectomorphism.
\end{definition}

We can without loss of generality always use the standard symplectic vector space: 
\begin{lemma} \label{lem:darboux}
    Every symplectic vector space \((V, \omega_V)\) over \(\F_p\) is symplectomorphic to \((\F_p^{2n}, \omega_n)\).
\end{lemma}

Just as the above group on \(\F_p \oplus \F_p^{2n}\) represents the Pauli group; the affine symplectomorphisms on the standard symplectic vector space represent the Clifford operators:
%
\begin{proposition}
  \label{lemma:clifford_symplectic}
  There is an isomorphism between the Clifford group modulo global phase,  \(\Proj(\Cliff)\),  and the group of affine symplectomorphisms on \((\F_p^{2n}, \omega_n)\).
\end{proposition}

\subsection{Affine Lagrangian relations}
Affine Lagrangian subspaces are \emph{the} fundamental notion in the symplectic representation:
\begin{enumerate}
\item
\begingroup
\relpenalty=10000   
\binoppenalty=10000 
\mathsurround=0pt   
the \emph{graph} of an affine symplectomorphism \(\phi : (X,\omega_X) \to (Y,\omega_Y)\)  is an affine Lagrangian subspace
 \({\Gr(\phi) \coloneqq \{(\vb{x}, \phi(\vb{x})) \mid x \in X\subseteq X
\oplus Y\} \subseteq (X\oplus Y,-\omega_X \oplus \omega_Y)}\);
\endgroup
\item
the composition of affine symplectomorphisms, and thus Clifford operators (by lemma~\ref{lemma:clifford_symplectic}), is compatible with the relational composition of their graphs \(\Gr(\psi)\circ \Gr(\phi) = \Gr(\psi \circ \phi)\);
\item
the action of Clifford operators on stabiliser states is compatible with the relational composition of their corresponding affine Lagrangian subspaces: if \(\phi\) is the affine symplectomorphism corresponding to a Clifford \(C\), and \(L_S\) is the affine Lagrangian subspace corresponding to a stabiliser group \(S\), then \(\Gr(\phi) \circ L_S = \phi(L_S) = L_{CSC^\dagger}\).
\end{enumerate}

%
In other words, following Weinstein \cite{Weinstein1982}, we adopt the motto:

\begin{center}
  \boxed{\emph{Everything is an affine Lagrangian relation!}}
\end{center}

This naturally leads to the following category which allows for these various notions to be seamlessly composed together:
\begin{definition}
  The \textbf{category, \(\ALR\), of affine Lagrangian relations} has:
  \begin{itemize}
    \item \textbf{objects}: symplectic vector spaces \((V,\omega_V)\);
    \item \textbf{morphisms \((V,\omega_V) \to (W,\omega_W)\)}: affine Lagrangian subspaces of \((V\oplus W, -\omega_V \oplus
    \omega_W)\);
    \item \textbf{identity and composition}: given by the diagonal relation and relational composition;
    \item \textbf{monoidal product}: given by the direct sum \((V \oplus W,
    \omega_V \oplus \omega_W)\);
    \item \textbf{monoidal unit}: given by the trivial symplectic vector space \(I\coloneqq (\F_p^{0},\omega_0)\);
    \item \textbf{dagger}: given by the relational converse;

    \item \textbf{\dag-compact structure}: the dual is given by  \((V,\omega_V)^* \coloneqq (V, -\omega_V)\) and the cup is given by the diagonal relation
    \(\eta_{(V,\omega_V)} \coloneqq \{ (\bullet, (\mathbf{v},\mathbf{v}))\}\). 
  \end{itemize}
\end{definition}

\begin{theorem}[\cite{comfort_graphical_2021, neretin_lectures_2011}]
  \label{theorem:stab_acr}
  There is an essentially surjective and full \dag-compact-closed functor 
  \(\rel:\Stab_p \to \ALR\).
  Quotienting by invertible scalars yields an equivalence 
  \(\Proj(\Stab_p) \simeq \ALR\).
\end{theorem}

%
By lemma~\ref{lem:darboux}, the affine Lagrangian relations between standard symplectic vector spaces of the form \((\F_p^{2n}, \omega_n)\) represent the basic operations of the stabiliser
theory.  However, the full stabiliser quantum theory is \emph{much richer} than what we have
described. For example, \emph{stabiliser codes} are fundamental
in quantum error correction, but they are not described by pure states.
Similarly, the measurement and classical control of stabiliser circuits, which we have not yet discussed, are essential for error correction.

It is worth mentioning that qubit stabiliser maps are not projectively equivalent to \(\ALR[\F_2]\). However, we can recover a projective equivalence \cite[Definition~3.3]{comfortthesis} by 
\begin{itemize}
    \item  restricting to the CSS stabiliser maps generated by the swap gate, the controlled not gate, the not gate, \(\ket{0}\) and \(\bra{0}\);

    \item restricting to the affine Lagrangian subspaces whose linear part splits into the direct sum of a linear subspace and its orthogonal complement \(L\oplus L^\perp\).
\end{itemize}





%
%
\section{Stabiliser codes and mixed states}
\label{section:mixed}
The stabiliser theory presented in Section~\ref{section:stabiliser}
is, from the perspective of quantum computation, fundamentally
limited: it is efficiently simulatable on a classical probabilistic
computer~\cite{gottesman_heisenberg_1998, aaronson2004improved}. Nevertheless,
the algebraic structure of the theory---specifically, the concept of stabiliser
subgroups---provides the foundation for quantum error correction (QEC). QEC
leverages these elements to encode and manipulate information in a way that
supports universal quantum computation, while retaining structural features
which permit fault tolerance.

A \textbf{stabiliser code} is a \emph{non-maximal} stabiliser group,
i.e. an Abelian subgroup \(G\) of \(\Pauli\) such that \(\xi(x) \cdot
1_{\mathcal{H}_p^{\otimes n}} \in S\) if and only if \(x \equiv  0 \mod
p\). Whereas a stabiliser group \(S\) uniquely determines a pure stabiliser
state \(\ket{S}\) up to global phase---a one-dimensional subspace of the
Hilbert space \(\Hp\)---a stabiliser code determines a higher-dimensional
subspace, the \textbf{codespace}:
\begin{equation}
  \mathcal{H}_G \coloneqq \{\ket{\phi} \in \Hp \mid s \ket{\phi}
  = \ket{\phi} \text{ for all } s \in G\} \subseteq \Hp.
\end{equation}
Elements of the stabiliser group impose linear constraints on the
codespace. Relaxing the number of constraints therefore yields a larger
subspace, while still enforcing sufficient symmetry to make it possible to
detect errors.

Semantically, it is natural to view a stabiliser code not just as a
subspace but as the completely mixed state on that subspace. To see this, given a stabiliser code \(G\), first observe that the orthonormal projection from \(\Hp\) onto \(\HS[G]\) is given by:
\begin{equation}
  \Pi_{G}
  \coloneqq 
  \frac{1}{\left|G\right|}
  \sum_{P\in G}
  P:\mathcal{H}_p^{\otimes n} \to \mathcal{H}_p^{\otimes n}.
\end{equation}
In the degenerate case, given a stabiliser state \([\exp(2\pi i \theta) \ket{S}]_{\theta \in [0,1)}\) with stabiliser group \(S\), this yields \(1\)-dimensional projector corresponding to the pure state \(\Pi_S \coloneqq \dyad{S}{S}\).
More generally, given a stabiliser code \(G\), the  trace-normalised map \(\rho_G \coloneqq
\Pi_G / \Tr(\Pi_G)\) is the uniformly \emph{mixed state} with support in 
\(\HS\subseteq \mathcal{H}_p^{\otimes n}\).  This allows stabiliser codes to be treated within the
framework of mixed-state quantum mechanics. 
In the following two subsections, we introduce two semantics for mixed stabiliser quantum mechanics by:
\begin{enumerate}
    \item \emph{restricting} the mixed processes to those built from stabiliser maps and the trace;
    \item \emph{generalising} the symplectic representation to affine coisotropic relations.
\end{enumerate}
Because the first semantics is given by restriction, it is semantically harder to grapple with.
On the other hand, the second semantics is novel, and much more natural to reason about stabiliser codes.

\subsection{Completely-positive maps between matrix algebras}
Mixed quantum processes are the result of exposing pure quantum processes  to their environment. Because all quantum processes used by quantum computers are exposed to the environment of their observer, mixed quantum theory is of central importance in quantum computing.  In the Schr\"odinger picture of quantum mechanics, the evolution of mixed quantum processes is described in terms of the action of certain completely positive maps between matrix algebras, i.e. between algebras of the form \(\mathcal{B}(\mathcal{H}) \cong \mathcal{H} \otimes \mathcal{H}^*\).
Starting from a \dag-CCC \(\mathsf{C}\), Selinger's CPM construction \(\CPM(\mathsf{C})\) yields an \emph{abstract} category of mixed processes in \(\mathsf{C}\), such that \(\CPM(\FHilb)\) captures the completely positive maps between finite-dimensional Hilbert spaces \cite{selinger_dagger_2007}. \(\CPM(\FHilb)\)  plays a similar role in quantum theory to how the Kleisli
categories \(\mathsf{Set}_{\mathscr{P}}\) and \(\mathsf{Meas}_{\mathscr{G}}\), respectively over the
power-set and Giry monads, are categorical semantics for nondeterministic
and probabilistic computation\footnote{In fact, the CPM construction is a
kind of dagger \emph{arrow} or \emph{promonad} on \(\mathsf{C}\) which generalises
the notion of Kleisli category of a monad \cite{heunen_reversible_2018}.}:
\begin{definition}[\cite{selinger_dagger_2007}]
    Given a \dag-CCC \(\mathsf{C}\), the \dag-CCC 
    \(\CPM(\mathsf{C})\) has:
    \begin{itemize}[nosep]
        \item \textbf{objects:} same as \(\mathsf{C}\);
        \item \textbf{morphisms \([f,S]:X \to Y\):} are equivalence classes of pairs
        \((f,S)\), where \(S\) is an object of \(\mathsf{C}\) and \(f:X
        \otimes S \to Y\) in \(\mathsf{C}\), modulo the equivalence relation
       \begin{equation*}
  (f,S) \sim (g,T)
  \iff
  \begin{tikzpicture}
	\begin{pgfonlayer}{nodelayer}
		\node [style=none] (0) at (5.275, 0.75) {};
		\node [style=none] (1) at (4.275, 1.375) {};
		\node [style=none] (2) at (4.275, 0.125) {};
		\node [style=none] (3) at (3.275, 1.375) {};
		\node [style=none] (4) at (3.275, 0.125) {};
		\node [style=none] (5) at (1.5, 1) {};
		\node [style=none] (6) at (4.275, 0.75) {};
		\node [style=none] (7) at (3.275, 1) {};
		\node [style=none] (8) at (3.275, 0.5) {};
		\node [style=none] (9) at (5.275, -0.75) {};
		\node [style=none] (10) at (4.275, -0.125) {};
		\node [style=none] (11) at (4.275, -1.375) {};
		\node [style=none] (12) at (3.275, -0.125) {};
		\node [style=none] (13) at (3.275, -1.375) {};
		\node [style=none] (14) at (1.5, -0.5) {};
		\node [style=none] (15) at (4.275, -0.75) {};
		\node [style=none] (16) at (3.275, -0.5) {};
		\node [style=none] (17) at (3.275, -1) {};
		\node [style=none] (18) at (3.775, 0.75) {$f$};
		\node [style=none] (19) at (3.775, -0.75) {$\overline f$};
		\node [style=wirelable] (20) at (4.9, 1) {$Y$};
		\node [style=wirelable] (21) at (4.9, -1) {$Y^*$};
		\node [style=wirelable] (22) at (1.9, 1.25) {$X$};
		\node [style=wirelable] (23) at (2.5, 0.625) {$S$};
		\node [style=none] (24) at (4.275, 1.375) {};
		\node [style=none] (25) at (4.275, 0.125) {};
		\node [style=none] (26) at (3.275, 1.375) {};
		\node [style=none] (27) at (3.275, 0.125) {};
		\node [style=none] (28) at (4.275, -1.375) {};
		\node [style=none] (29) at (4.275, -0.125) {};
		\node [style=none] (30) at (3.275, -1.375) {};
		\node [style=none] (31) at (3.275, -0.125) {};
		\node [style=none] (32) at (5.275, 1.625) {};
		\node [style=none] (33) at (1.525, 1.625) {};
		\node [style=none] (34) at (1.525, -1.625) {};
		\node [style=none] (35) at (5.275, -1.625) {};
		\node [style=wirelable] (36) at (1.9, -0.75) {$X^*$};
		\node [style=wirelable] (37) at (2.5, -1.125) {$S^*$};
	\end{pgfonlayer}
	\begin{pgfonlayer}{edgelayer}
		\draw [style=background] (35.center)
			 to (32.center)
			 to (33.center)
			 to (34.center)
			 to cycle;
		\draw [style=thick, in=0, out=180] (7.center) to (5.center);
		\draw [style=thick] (6.center) to (0.center);
		\draw [style=thick] (16.center) to (14.center);
		\draw [style=thick] (15.center) to (9.center);
		\draw [style=thick, bend left=90, looseness=2.25] (17.center) to (8.center);
		\draw [style=white] (2.center)
			 to (1.center)
			 to (3.center)
			 to (4.center)
			 to cycle;
		\draw [style=white] (10.center)
			 to (11.center)
			 to (13.center)
			 to (12.center)
			 to cycle;
		\draw [style={thick_bbox}] (25.center)
			 to (24.center)
			 to (26.center)
			 to (27.center)
			 to cycle;
		\draw [style={thick_bbox}] (28.center)
			 to (30.center)
			 to (31.center)
			 to (29.center)
			 to cycle;
	\end{pgfonlayer}
\end{tikzpicture}
  =
  \begin{tikzpicture}
	\begin{pgfonlayer}{nodelayer}
		\node [style=none] (0) at (5.275, 0.75) {};
		\node [style=none] (1) at (4.275, 1.375) {};
		\node [style=none] (2) at (4.275, 0.125) {};
		\node [style=none] (3) at (3.275, 1.375) {};
		\node [style=none] (4) at (3.275, 0.125) {};
		\node [style=none] (5) at (1.5, 1) {};
		\node [style=none] (6) at (4.275, 0.75) {};
		\node [style=none] (7) at (3.275, 1) {};
		\node [style=none] (8) at (3.275, 0.5) {};
		\node [style=none] (9) at (5.275, -0.75) {};
		\node [style=none] (10) at (4.275, -0.125) {};
		\node [style=none] (11) at (4.275, -1.375) {};
		\node [style=none] (12) at (3.275, -0.125) {};
		\node [style=none] (13) at (3.275, -1.375) {};
		\node [style=none] (14) at (1.5, -0.5) {};
		\node [style=none] (15) at (4.275, -0.75) {};
		\node [style=none] (16) at (3.275, -0.5) {};
		\node [style=none] (17) at (3.275, -1) {};
		\node [style=none] (18) at (3.775, 0.75) {$g$};
		\node [style=none] (19) at (3.775, -0.75) {$\overline g$};
		\node [style=wirelable] (20) at (4.9, 1) {$Y$};
		\node [style=wirelable] (21) at (4.9, -1) {$Y^*$};
		\node [style=wirelable] (22) at (1.9, 1.25) {$X$};
		\node [style=wirelable] (23) at (2.5, 0.625) {$T$};
		\node [style=none] (24) at (4.275, 1.375) {};
		\node [style=none] (25) at (4.275, 0.125) {};
		\node [style=none] (26) at (3.275, 1.375) {};
		\node [style=none] (27) at (3.275, 0.125) {};
		\node [style=none] (28) at (4.275, -1.375) {};
		\node [style=none] (29) at (4.275, -0.125) {};
		\node [style=none] (30) at (3.275, -1.375) {};
		\node [style=none] (31) at (3.275, -0.125) {};
		\node [style=none] (32) at (5.275, 1.625) {};
		\node [style=none] (33) at (1.525, 1.625) {};
		\node [style=none] (34) at (1.525, -1.625) {};
		\node [style=none] (35) at (5.275, -1.625) {};
		\node [style=wirelable] (36) at (1.9, -0.75) {$X^*$};
		\node [style=wirelable] (37) at (2.5, -1.125) {$T^*$};
	\end{pgfonlayer}
	\begin{pgfonlayer}{edgelayer}
		\draw [style=background] (35.center)
			 to (32.center)
			 to (33.center)
			 to (34.center)
			 to cycle;
		\draw [style=thick, in=0, out=180] (7.center) to (5.center);
		\draw [style=thick] (6.center) to (0.center);
		\draw [style=thick] (16.center) to (14.center);
		\draw [style=thick] (15.center) to (9.center);
		\draw [style=thick, bend left=90, looseness=2.25] (17.center) to (8.center);
		\draw [style=white] (2.center)
			 to (1.center)
			 to (3.center)
			 to (4.center)
			 to cycle;
		\draw [style=white] (10.center)
			 to (11.center)
			 to (13.center)
			 to (12.center)
			 to cycle;
		\draw [style={thick_bbox}] (25.center)
			 to (24.center)
			 to (26.center)
			 to (27.center)
			 to cycle;
		\draw [style={thick_bbox}] (28.center)
			 to (30.center)
			 to (31.center)
			 to (29.center)
			 to cycle;
	\end{pgfonlayer}
\end{tikzpicture}
  \quad \text{where} \quad
  \begin{tikzpicture}
	\begin{pgfonlayer}{nodelayer}
		\node [style=none] (0) at (-1.25, 0.675) {};
		\node [style=none] (1) at (1.25, 0.675) {};
		\node [style=none] (2) at (1.25, -0.625) {};
		\node [style=none] (3) at (-1.25, -0.625) {};
		\node [style=none] (5) at (-1.25, 0) {};
		\node [style=none] (6) at (1.25, 0) {};
		\node [style=wirelable] (7) at (1, 0.25) {$Y^*$};
		\node [style=gbox] (8) at (0, 0) {$\overline{f}$};
		\node [style=wirelable] (9) at (-0.875, 0.25) {$X^*$};
	\end{pgfonlayer}
	\begin{pgfonlayer}{edgelayer}
		\draw [style=background] (3.center)
			 to (0.center)
			 to (1.center)
			 to (2.center)
			 to cycle;
		\draw [style=thick] (5.center) to (6.center);
	\end{pgfonlayer}
\end{tikzpicture}
  \coloneqq
  \begin{tikzpicture}
	\begin{pgfonlayer}{nodelayer}
		\node [style=none] (0) at (0.5, 1.05) {};
		\node [style=none] (1) at (4, 1.05) {};
		\node [style=none] (2) at (4, -1) {};
		\node [style=none] (3) at (0.5, -1) {};
		\node [style=none] (4) at (1.75, 0) {};
		\node [style=none] (5) at (2.75, 0) {};
		\node [style=wirelable] (6) at (3.75, 0.25) {$Y^*$};
		\node [style=gbox] (7) at (2.25, 0) {$f^\dagger$};
		\node [style=wirelable] (8) at (0.875, 0.25) {$X^*$};
		\node [style=none] (9) at (1, 0) {};
		\node [style=none] (10) at (3.5, 0) {};
		\node [style=none] (11) at (4, 0) {};
		\node [style=none] (12) at (0.5, 0) {};
		\node [style=none] (13) at (1.75, -0.75) {};
		\node [style=none] (14) at (2.75, -0.75) {};
		\node [style=none] (15) at (1.75, 0.75) {};
		\node [style=none] (16) at (2.75, 0.75) {};
	\end{pgfonlayer}
	\begin{pgfonlayer}{edgelayer}
		\draw [style=background] (3.center)
			 to (0.center)
			 to (1.center)
			 to (2.center)
			 to cycle;
		\draw [style=thick] (12.center)
			 to (9.center)
			 to [in=-180, out=0] (13.center)
			 to (14.center)
			 to [bend right=90, looseness=1.50] (5.center)
			 to (4.center)
			 to [bend left=90, looseness=1.50] (15.center)
			 to (16.center)
			 to [in=180, out=0] (10.center)
			 to (11.center);
	\end{pgfonlayer}
\end{tikzpicture}.
\end{equation*}

        \item \textbf{all other \dag-compact-closed structure:} inherited from \(\mathsf{C}\).
    \end{itemize}
\end{definition}

There is a canonical functor \(\iota:\mathsf{C}\to \CPM(\mathsf{C})\) which
sends ``pure'' morphisms \(f:X\to Y\) to their ``doubled'' counterpart \([f,
I]:X\to Y\).  This means that for each object \(X\) the CPM construction is
adding a new state \([1_X\otimes 1_I, I]:I\to X\) that connects both halves
of this doubling.
For the example of \(\mathsf{C}\coloneqq \FHilb\), this ``doubling''
is analogous to the usual map embedding pure states into mixed states:
\(\ket{\psi} \mapsto \dyad{\psi}\). The new morphism is interpreted as adding
the (unnormalised) maximally mixed state to pure-state quantum mechanics:

\begin{theorem}[{\cite[Ex. 4.21]{selinger_dagger_2007}}]
    \(\CPM(\FHilb)\) is equivalent to the CCC of \emph{completely-positive}
    (CP) maps between matrix algebras \(\mathcal{B}(\mathcal{H})\), for
    all finite-dimensional Hilbert spaces \(\mathcal{H}\) in \(\FHilb\).
    In particular,  \([ 1_\mathcal{H} /\dim\mathcal{H},\mathcal{H}]:
    \C \to \mathcal{H}\) represents the \textbf{maximally mixed state}
    \(\rho_{\mathsf{max}}\) on \(\mathcal{B}(\mathcal{H})\).
\end{theorem}
Moreover, by restricting to the stabiliser maps, we have immediately that:
\begin{lemma}
  There is a faithful \dag-compact functor \(\CPM(\Stab_p) \rightarrowtail \CPM(\FHilb)\).
\end{lemma}

\begin{example}
    The projection \(\Pi_G :\Hp \to \Hp\) onto the codespace of a stabiliser code \(G\) is a state in \(\CPM(\Stab_p)\).  Moreover, Clifford operators \(C\)  in \( \Cliff \rightarrowtail \Stab_p\to \CPM(\Stab_p)\) act on these projectors by conjugation \(C^\dag \Pi_G C =  \Pi_{C^\dag G C}\) as expected.  In particular, given a stabiliser state \(\ket{G}\), we have \( C^\dag\dyad{G}{G}C = \dyad{C^\dag G C}{C^\dag G C}\).    
\end{example}

To restrict to \emph{physical} processes, we impose the additional normalisation constraint.
\begin{definition}
  Given any \(X \in \mathsf{C}\), denote the \textbf{abstract trace} in  by the morphism  \(\Tr_X\coloneqq [\epsilon_X \circ \swap_{X,X^*}, X^*]: X\to I\) in \(\CPM(\mathsf{C})\).
  A morphism \([f,S] : X \to Y\) in \(\CPM(\mathsf{C})\) is \textbf{causal} if and only if \(\Tr_Y[f,S] = \Tr_X\), diagrammatically:
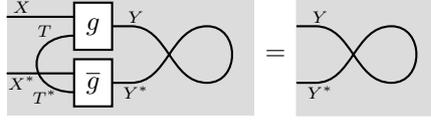
  \begin{center}
    \begin{tikzpicture}
	\begin{pgfonlayer}{nodelayer}
		\node [style=none] (0) at (4.775, 0.75) {};
		\node [style=none] (1) at (4.275, 1.375) {};
		\node [style=none] (2) at (4.275, 0.125) {};
		\node [style=none] (3) at (3.275, 1.375) {};
		\node [style=none] (4) at (3.275, 0.125) {};
		\node [style=none] (5) at (1.5, 1) {};
		\node [style=none] (6) at (4.275, 0.75) {};
		\node [style=none] (7) at (3.275, 1) {};
		\node [style=none] (8) at (3.275, 0.5) {};
		\node [style=none] (9) at (4.775, -0.75) {};
		\node [style=none] (10) at (4.275, -0.125) {};
		\node [style=none] (11) at (4.275, -1.375) {};
		\node [style=none] (12) at (3.275, -0.125) {};
		\node [style=none] (13) at (3.275, -1.375) {};
		\node [style=none] (14) at (1.5, -0.5) {};
		\node [style=none] (15) at (4.275, -0.75) {};
		\node [style=none] (16) at (3.275, -0.5) {};
		\node [style=none] (17) at (3.275, -1) {};
		\node [style=none] (18) at (3.775, 0.75) {$g$};
		\node [style=none] (19) at (3.775, -0.75) {$\overline g$};
		\node [style=wirelable] (20) at (4.9, 1) {$Y$};
		\node [style=wirelable] (21) at (4.9, -1) {$Y^*$};
		\node [style=wirelable] (22) at (1.9, 1.25) {$X$};
		\node [style=wirelable] (23) at (2.5, 0.625) {$T$};
		\node [style=none] (24) at (4.275, 1.375) {};
		\node [style=none] (25) at (4.275, 0.125) {};
		\node [style=none] (26) at (3.275, 1.375) {};
		\node [style=none] (27) at (3.275, 0.125) {};
		\node [style=none] (28) at (4.275, -1.375) {};
		\node [style=none] (29) at (4.275, -0.125) {};
		\node [style=none] (30) at (3.275, -1.375) {};
		\node [style=none] (31) at (3.275, -0.125) {};
		\node [style=none] (32) at (8.025, 1.625) {};
		\node [style=none] (33) at (1.525, 1.625) {};
		\node [style=none] (34) at (1.525, -1.625) {};
		\node [style=none] (35) at (8.025, -1.625) {};
		\node [style=wirelable] (36) at (1.9, -0.75) {$X^*$};
		\node [style=wirelable] (37) at (2.5, -1.125) {$T^*$};
		\node [style=none] (38) at (6.775, -0.75) {};
		\node [style=none] (39) at (6.775, 0.75) {};
		\node [style=none] (72) at (12.875, 1.625) {};
		\node [style=none] (73) at (9.125, 1.625) {};
		\node [style=none] (74) at (9.125, -1.625) {};
		\node [style=none] (75) at (12.875, -1.625) {};
		\node [style=none] (80) at (8.575, 0) {\(=\)};
		\node [style=none] (81) at (9.625, 0.75) {};
		\node [style=none] (82) at (9.125, 0.75) {};
		\node [style=none] (83) at (9.625, -0.75) {};
		\node [style=none] (84) at (9.125, -0.75) {};
		\node [style=wirelable] (85) at (9.75, 1) {$Y$};
		\node [style=wirelable] (86) at (9.75, -1) {$Y^*$};
		\node [style=none] (87) at (11.625, -0.75) {};
		\node [style=none] (88) at (11.625, 0.75) {};
	\end{pgfonlayer}
	\begin{pgfonlayer}{edgelayer}
		\draw [style=background] (35.center)
			 to (32.center)
			 to (33.center)
			 to (34.center)
			 to cycle;
		\draw [style=thick, in=0, out=180] (7.center) to (5.center);
		\draw [style=thick] (15.center)
			 to (9.center)
			 to [in=180, out=0] (39.center)
			 to [bend left=90, looseness=1.50] (38.center)
			 to [in=360, out=180] (0.center)
			 to (6.center);
		\draw [style=thick] (16.center) to (14.center);
		\draw [style=thick, bend left=90, looseness=2.25] (17.center) to (8.center);
		\draw [style=white] (2.center)
			 to (1.center)
			 to (3.center)
			 to (4.center)
			 to cycle;
		\draw [style=white] (10.center)
			 to (11.center)
			 to (13.center)
			 to (12.center)
			 to cycle;
		\draw [style={thick_bbox}] (25.center)
			 to (24.center)
			 to (26.center)
			 to (27.center)
			 to cycle;
		\draw [style={thick_bbox}] (28.center)
			 to (30.center)
			 to (31.center)
			 to (29.center)
			 to cycle;
		\draw [style=background] (74.center)
			 to (75.center)
			 to (72.center)
			 to (73.center)
			 to cycle;
		\draw [style=background] (75.center) to (72.center);
		\draw [style=background] (72.center) to (73.center);
		\draw [style=background] (73.center) to (74.center);
		\draw [style=background] (74.center) to (75.center);
		\draw [style=thick] (84.center)
			 to (83.center)
			 to [in=180, out=0] (88.center)
			 to [bend left=90, looseness=1.50] (87.center)
			 to [in=360, out=180] (81.center)
			 to (82.center);
	\end{pgfonlayer}
\end{tikzpicture}
  \end{center}
  Denote the symmetric monoidal subcategory of causal morphisms  by \(\Caus(\CPM(\mathsf{C}))\).
\end{definition}
Concretely, the morphisms \(\Tr_{\mathcal H}\) in \(\CPM(\FHilb)\), are given by the linear-algebraic trace in \(\FHilb\).
In the language of operator algebras:
\begin{corollary}
  \(\Caus(\CPM(\FHilb))\) is equivalent to the SMC of completely-positive \textbf{trace-preserving} (CPTP) maps between matrix algebras \(\mathcal{B}(\mathcal{H})\) for all \(\mathcal{H} \in \FHilb\).
  In particular \(\Tr_{\mathcal{H}}:\mathcal{H}\to \C\) represents the matrix trace on \(\mathcal{B}(\mathcal{H})\).
\end{corollary}
In other words, these are CP maps which preserve the trace norm, in analogy to how Markov processes preserve the \(L^1\) norm.
\begin{example}
The trace-normalisation \(\rho_G = \Pi_G/\Tr(\Pi_G)\) of the projector \(\Pi_G\) onto the code space of a stabiliser code \(G\) is completely positive and trace preserving; whereas without the normalisation factor, the projector \(\Pi_G\) is only trace-preserving when \(\Tr(\Pi_G)=1\), ie. when \(G\) is a stabiliser group.  
\end{example}

\subsection{Mixed stabiliser maps as affine coisotropic relations}
In this subsection we apply the CPM construction to the category of stabiliser maps and the category affine Lagrangian relations, giving two  compositional semantics for stabiliser codes: one in terms of mixed-state stabiliser quantum mechanics and the other in terms of symplectic linear algebra.  In particular, we show that applying this construction to the category affine Lagrangian relations produces the category of affine \emph{coisotropic} relations:
\begin{definition}
  The \dag-CCC, \(\ACR\), of \textbf{affine coisotropic relations}
  has the same structure as \(\ALR\), where now the morphisms \((V,\omega_V) \to
  (W,\omega_W)\) are affine coisotropic subspaces of \((V\oplus W,-\omega_V \oplus
  \omega_W)\).
\end{definition}

It is well-understood that “phaseless” stabiliser codes are in
bijection with isotropic subspaces~\cite{gross_finite_2005}, and hence
also with coisotropic subspaces via the symplectic complement. However,
we shall see that once the Pauli phases \(\chi(a)\) are reintroduced as in
equation~\eqref{eq:symplectic_presentation}, this second bijection breaks down:
in the following paragraphs we show that phased stabiliser codes correspond
exactly to affine coisotropic subspaces but \emph{not} to affine isotropic
subspaces. Thus, affine coisotropic relations are the correct algebraic setting
for the stabiliser theory with non-maximal stabiliser groups. To prove this
correspondence, we first must understand the basic structure of coisotropic
subspaces, starting with the symplectic analogue of the maximally mixed state:

\begin{example}
    The total subspace can be regarded as an affine coisotropic relation: 
    \begin{equation}\im_{(V,\omega_V)}\coloneqq \{(0,\vb{v}) \ | \ \forall \vb{v} \in V \}:(\F_p^0, \omega_0)\rightarrowtail (V,\omega_V)\end{equation}
\end{example}

We use the notation \(\im_{(V,\omega_V)}\) because postcomposition with an affine Lagrangian, or affine coisotropic relation \(R:(V,\omega_V)\to (W,\omega_W)\) is identified with the set-theoretic image:
\begin{equation}
     R \circ \im_{(V,\omega_V)} = \{(0, \vb{w}) \ | \ \exists \vb{v}: (\vb{v}, \vb{w}) \in R\} = \Zp^0 \oplus \im(R) \cong \im(R)
\end{equation}

Adding the image as a generator to \(\ALR\) yields \(\ACR\). The proof of this fact relies on the following result, which mirrors the construction of a Stinespring dilation of completely positive maps between finite-dimensional Hilbert spaces, but in the symplectic setting:
\begin{proposition}
\label{proposition:isotropic_kernel}
  Every non-empty affine coisotropic subspace of \((\F_p^{2n}, \omega_n)\) of dimension \(2n-k\) is the image
  of an affine Lagrangian isometry \((\F_p^{2(n-k)},\omega_{n-k}) \rightarrowtail (\F_p^{2n},\omega_n)\).
\end{proposition}
\begin{proof}
  We prove the claim for linear Lagrangian coisometries and coisotropic linear subspaces, after which the affine generalisation follows immediately.

  Take a linear coisotropic subspace \(S\subseteq (\F_p^{2n},\omega_n)\) of
  dimension \(2n-k\). Then its symplectic complement \(S^{\omega_n}\subseteq
  (\F_p^{2n},\omega_n)\) is an isotropic subspace of dimension
  \(k\). Consider a basis of \(S^{\omega_n}\) which by \cite[Exercise
  7.2.2-(a)]{rudolph_linear_2013} extends to a symplectic basis of
  \((\F_p^{2n}, \omega_n)\). By \cite[Corollary 7.1.5]{rudolph_linear_2013},
  this yields a symplectomorphism \(\phi : (\F_p^{2n}, \omega_n)
  \smash{\xrightarrow{\cong}} (\F_p^{2n}, \omega_n)\) such that \(\phi(S^\omega)
  = \{(\vb{x}, \vb{0}_{2n-k}) \mid \vb{x} \in \F_p^k\}\). Remark that \(D
  \coloneqq \{(\vb{x}, \vb{0}_k) \mid \vb{x} \in \F_p^k\}\subseteq(\F_p^{2k},
  \omega_k)\) is a Lagrangian subspace, inducing a Lagrangian coisometry \(E\coloneqq
  D \oplus \F_p^0 : (\F_p^{2k}, \omega_k)\twoheadrightarrow (\F_p^0,
  \omega_0)\). Therefore the following composite of Lagrangian coisometries
  is itself a Lagrangian coisometry:
  \begin{equation}
    C \coloneqq (E \oplus 1_{(\F_p^{2(n-k)}, \omega_{n-k})})\circ \Gr(\phi):(\F_p^{2n},\omega_n)\twoheadrightarrow (\F_p^{2(n-k)},\omega_{n-k}).
  \end{equation}
  Taking the kernel of this Lagrangian coisometry, we have:
  \begin{align}
  \F_p^0\oplus \ker{C} &= C^\dag(\{\vb{0}_{2(n-k)}\}) 
          = (\Gr(\phi)^\dag\circ (E^\dag \oplus 1_{(\F_p^{2(n-k)}, \omega_{n-k})}))(\{\vb{0}_{2(n-k)}\}) \\
          &= \Gr(\phi)^\dag\circ \{(\bullet, (\vb{d}, \vb{0}_{2(n-k)}))\ | \ \vb{d} \in D \}\\
          &= \Gr(\phi)^\dag\circ \{(\bullet, (\vb{x}, \vb{0}_{2n-k}))\ | \ \vb{x} \in \F_p^k \}\\
          &= \{(\bullet, \phi^{-1}(D\oplus \{ \vb{0}_{2n-k}\}))\}
          = \{(\bullet, \vb{s})  \ | \ \vb{s} \in S^{\omega_n}\} = \F_p^0 \oplus S^{\omega_n}.
  \end{align}
  Finally, by taking the double symplectic complement of \(\F_p^0 \oplus S\), we have the desired result
  \begin{align}
    \F_p^0 \oplus S &= ((\F_p^0 \oplus S)^{-\omega_0\oplus\omega_n})^{-\omega_0\oplus\omega_n} = (\F_p^0\oplus \ker{C})^{-\omega_0\oplus\omega_n}\\
    &= \F_p^0\oplus \im(C^\dag)=  C^\dag \circ \im_{(\F_p^{2(n-k)},\omega_{n-k})}. 
  \end{align}
\end{proof}

\begin{theorem}
  There is a \dag-compact isomorphism \(\CPM(\LR) \cong \CR\) sending:
  \begin{center}
        \([f,(S, \omega_S)]\quad \longmapsto\quad f\circ (1 \oplus \im_{(S,\omega_S)}) \)
  \end{center}
\end{theorem}
\begin{proof}
  This assignment is clearly functorial and identity-on-objects, and preserves the \dag-compact-closed structure. Moreover, since both \(\CPM(\LR)\) and \(\CR\) are compact-closed, it suffices to
  prove that the states in both categories are in canonical bijection. We already have surjectivity by  proposition~\ref{proposition:isotropic_kernel}, so that all we need to prove is injectivity.
  
  Given Lagrangian relations \(L : (S,\omega_S) \to (V,\omega_V)\) and \(M : (T,\omega_T) \to (V,\omega_V)\) such that
  \(\im(L) \neq \im(M)\), then
  \begin{equation}
    \vb{v} \in
    \im(L)
    \qq{if and only if}
    \begin{bmatrix} \vb{v} \\ \vb{v} \end{bmatrix} \in \interp{{\begin{tikzpicture}
	\begin{pgfonlayer}{nodelayer}
		\node [style=none] (0) at (1.25, 1.125) {};
		\node [style=none] (1) at (-1, 1.125) {};
		\node [style=none] (2) at (-1, -1.125) {};
		\node [style=none] (3) at (1.25, -1.125) {};
		\node [style={thick_box}] (4) at (0, 0.5) {$L$};
		\node [style={thick_box}] (5) at (0, -0.5) {$\overline{L}$};
		\node [style=none] (6) at (1.25, 0.5) {};
		\node [style=none] (7) at (1.25, -0.5) {};
		\node [style=wirelable] (8) at (0.825, 0.75) {$V$};
		\node [style=wirelable] (9) at (0.95, -0.25) {$V^*$};
		\node [style=none] (12) at (-0.25, -0.5) {};
		\node [style=none] (13) at (-0.25, 0.5) {};
	\end{pgfonlayer}
	\begin{pgfonlayer}{edgelayer}
		\draw [style=background] (2.center)
			 to (3.center)
			 to [in=270, out=90] (0.center)
			 to (1.center)
			 to cycle;
		\draw [style=thick] (4) to (6.center);
		\draw [style=thick] (5) to (7.center);
		\draw [style=thick] (5)
			 to (12.center)
			 to [bend right=270, looseness=1.50] (13.center)
			 to (4);
	\end{pgfonlayer}
\end{tikzpicture}}}
    = \left\{\begin{bmatrix} \vb{v} \\ \vb{w} \end{bmatrix} \;\middle|\; \exists \vb{s}: \begin{array}{l} (\vb{s},\vb{v}) \in L \\ (\vb{s},\vb{w}) \in \overline{L} \end{array} \right\}.
  \end{equation}
  However,  assumption there is some \(\vb{v} \in V\) such that \(\vb{v} \in
  \im(L)\) and \(\vb{v} \notin \im(M)\), so that
  \([L,(S,\omega_S)] \neq [M, (T, \omega_T)]\).
\end{proof}

\begin{corollary}
  There is a \dag-compact isomorphism \(\CPM(\ALR) \cong \ACR\) sending:
  \begin{center}
        \([f,(S, \omega_S)]\quad \longmapsto\quad f\circ (1 \oplus \im_{(S,\omega_S)}) \)
  \end{center}
\end{corollary}
\begin{proof}
  This follows immediately by composing with translations (which are invertible).
\end{proof}
  
This allows us to conclude that affine coisotropic relations are a projective representation of mixed stabiliser maps; i.e. the processes generated by stabiliser maps in addition to the maximally mixed state:
\begin{corollary}
  \label{corollary:cpm_acr}
  There is an essentially surjective, full, \dag-compact functor \(\rel:\CPM(\Stab_p)\twoheadrightarrow \ACR\).
  Quotienting by scalars yields an equivalence \( \Proj(\CPM(\Stab_p))\cong \ACR\).
\end{corollary}
\begin{proof}
  This follows immediately from the equivalence \(\ACR \cong \CPM(\ALR) \simeq \CPM(\Proj(\Stab_p))\) and observing that, in the case of \(\Stab_p\), we obtain the same category if we quotient by scalars before or after applying the CPM construction.
\end{proof}

To include mixed states and stabiliser codes in our semantics,
we update our motto:
\begin{center}
  \boxed{\emph{Everything is an affine \underline{coisotropic} relation!}}
\end{center}
\begin{example} \label{ex:channel_relations}
  The following quantum channels are represented by 
  \(\rel(\CPM(\Stab_p))\):
  \begin{itemize}[nosep]
    \item  \emph{maximally mixed state}:
      \(\rel\!\big(\rho_{\mathsf{max}}\coloneqq \tfrac{1}{p^{n}}\sum \lvert j\rangle\!\langle j\rvert\big)
      = \im_{(\Zp^{2n}, \omega_n)}\);
    \item  \emph{matrix trace}: 
      \(\rel\!\big(\Tr\coloneqq \sum \langle j \lvert (-) \rvert j \rangle\big)
      = \im_{(\Zp^{2n}, \omega_n)}^\dagger\);
    \item  \emph{completely depolarising channel}: 
      \(\rel\!\big(\tfrac{1}{p^{n}}\sum \lvert j\rangle\!\langle k\rvert\,(-)\,\lvert k\rangle\!\langle j\rvert\big)
      = \im_{(\Zp^{2n}, \omega_n)} \circ \im_{(\Zp^{2n}, \omega_n)}^\dagger\);
    \item \emph{\(Z\)-decoherence}: 
      \(\rel\!\big(\mathcal{E}_Z\coloneqq \sum \lvert j\rangle\!\langle j\rvert\,(-)\,\lvert j\rangle\!\langle j\rvert\big)
      = \Big\{\Big(\bigl[\begin{smallmatrix} x\\ z\end{smallmatrix}\bigr],\bigl[\begin{smallmatrix} x\\ z'\end{smallmatrix}\bigr]\Big) \  \Big| \ x,z,z' \in \mathbb{F}_p \Big\}
      \eqqcolon \mathscr{E}_Z\).
  \end{itemize}
\end{example}

Just as we restrict completely positive maps to be trace preserving, to capture the physical processes, in the symplectic picture, we restrict ourselves to the total relations:
\begin{definition}
  A relation \(R:X\to Y\) with converse \(R^\dagger\) is \textbf{total} when \(\im(R^\dagger) = X\).
  Given a category \(\mathsf{C}\) of relations, let \(\Total(\mathsf{C})\) denote the subcategory of total maps.
\end{definition}

The correspondence \(\CPM(\Stab_p)\rightarrowtail \ACR\) restricts to an equivalence, \emph{without quotienting by scalars} when  considering only the physical proceses:
\begin{proposition}
  \label{proposition:causal_total_cpm}
  The functor \(\operatorname{Rel}:\CPM(\Stab_p) \to \ACR\) restricts to a symmetric monoidal equivalence \(\Caus(\CPM(\Stab_p))\simeq \Total(\ACR)\) making the following diagram of symmetric monoidal categories commute:
\[\begin{tikzcd}[row sep=.25cm, cells={nodes={inner sep=1pt, outer sep=0pt}}]
	{\Total(\ACR)} && \ACR \\[-.1cm]
	{\Caus(\CPM(\Stab_p))} & {\CPM(\Stab_p)} & {\Proj(\CPM(\Stab_p))} \\
	{\Caus(\CPM(\FHilb))} & {\CPM(\FHilb)} & {\Proj(\CPM(\FHilb)} \\
    \begin{array}{c} \begin{matrix}\text{CPTP maps between}\\[-2pt]\text{matrix algebras}\end{matrix} \end{array} & \begin{array}{c} \begin{matrix}\text{CP maps between}\\[-2pt]\text{matrix algebras}\end{matrix} \end{array} & \begin{array}{c} \Proj\left(\begin{matrix}\text{CP maps between}\\[-2pt]\text{matrix algebras}\end{matrix}\right) \end{array}
	\arrow[tail, from=1-1, to=1-3]
	\arrow["\!\simeq"{marking, allow upside down}, draw=none, from=1-1, to=2-1]
	\arrow["\!\simeq"{marking, allow upside down}, draw=none, from=1-3, to=2-3]
	\arrow[tail, from=2-1, to=2-2]
	\arrow[tail, from=2-1, to=3-1]
	\arrow[two heads, from=2-2, to=2-3]
	\arrow[tail, from=2-3, to=3-3]
	\arrow[tail, from=3-1, to=3-2]
	\arrow[two heads, from=3-2, to=3-3]
    \arrow[tail, from=2-2, to=3-2]
    \arrow["\!\simeq"{marking, allow upside down}, draw=none, from=3-1, to=4-1]
    \arrow["\!\simeq"{marking, allow upside down}, draw=none, from=3-2, to=4-2]
    \arrow["\!\simeq"{marking, allow upside down}, draw=none, from=3-3, to=4-3]
    \arrow[tail, from=4-1, to=4-2]
    \arrow[two heads, from=4-2, to=4-3]
\end{tikzcd}\]
\end{proposition}
\begin{proof}
  By corollary~\ref{corollary:cpm_acr}, \(\operatorname{Rel}:\Caus(\CPM(\Stab_p)) \to \Total(\ACR)\) is an essentially surjective, full, monoidal functor, making the diagram commute. It remains to prove faithfulness.
  Take two maps \([f,S],[g,T]:X\to Y\) in \(\CPM(\Stab)\) such that \([g,T] = \lambda \cdot [f,S]\) some
  \(\lambda\neq 0\). Then \(\Tr_Y[g,T]= \lambda\Tr_Y[f,S] = \lambda \Tr_X\).  
  Therefore \([g,T]\) is causal iff \(\lambda = 1\) i.e. \([g,T] =
  [f,S] \), thus, each projective equivalence class of morphisms \(\Caus(\CPM(\Stab_p))\) contains at most one representative. Therefore, the equivalence \(\ACR \simeq
  \Proj(\CPM(\Stab_p))\) uniquely lifts  along \(\Proj\) on causal maps.
\end{proof}
Therefore, the trace preserving mixed stabiliser maps are in canonical bijection with the total affine coisotropic relations.
Indeed, transporting proposition~\ref{proposition:isotropic_kernel} along \( \Total(\ACR)\cong \Caus(\CPM(\Stab_p))\), we find that  every completely positive trace-preserving mixed stabiliser map admits a Stinespring dilation into a stabiliser isometry.
Note that all of the mixed stabiliser maps in example~\ref{ex:channel_relations} are trace-preserving, and their affine coisotropic counterparts are total.

Moreover, this representation restricts to the completely positive trace-preserving mixed qubit CSS maps, considering total affine coisotropic whose linear component splits into the direct sum of two linear subspaces.

\subsection{Stabiliser quantum error correction via affine coisotropic relations}

Stabiliser quantum mechanics is the cornerstone of finite-dimensional quantum error correction \cite{gottesman_stabilizer_1997, Ashikhmin2001}. In this subsection will give a dictionary for how the fundamental constructions in stabiliser QEC  can be represented in the symplectic language.  Then we give a compositional account for the coarse-graining of stabiliser quantum error correction codes.

\paragraph*{The symplectic formulation of stabiliser quantum error correction}

To understand the connection between the symplectic geometry and quantum
error correction, we revisit the notion of a stabiliser code. The stabiliser
code associated to a non-empty affine coisotropic subspace \(C=L+\vb{a} \subseteq
(\Zp^{2n},\omega_n)\), is the Abelian subgroup
\begin{equation}
  G_C\coloneqq \{\pi\bigl(\omega_n(\vb{a},\vb b),\vb{b}) \ \mid \ \vb{b} \in L^{\omega_n}\} \subseteq \Pauli,
\end{equation}
moreover the projector onto the codespace \(  \mathcal{H}_{G_C}\cong \mathcal{H}_p^{\otimes (\dim(L)-n)} \subseteq \mathcal{H}_p^{\otimes n}\) is given by:
\begin{equation}
  \Pi_{G_C}
  \coloneqq 
  \frac{1}{\left|L^{\omega_n}\right|}
  \sum_{\vb{b}\in L^{\omega_n}}
  \pi\bigl(\omega_n(\vb{a},\vb b),\vb{b} \bigr):\mathcal{H}_p^{\otimes n} \to \mathcal{H}_p^{\otimes n}.
\end{equation}

By convention the empty affine coisotropic subspace is identified with the  Abelian subgroup \(G_\emptyset \coloneqq \langle e^{i \frac{2\pi}{p}} 1_{\mathcal{H}_p^{\otimes n}}\rangle \subseteq \Pauli\), so that \(\Pi_{G_{\emptyset}} \coloneqq 0\). This is chosen so that for any (possibly empty) affine coisotropic subspace \(S\subseteq  (\Zp^{2n},\omega_n)\), we have that \(\rel(\Pi_{G_C}) = S\).

An \textbf{encoder} for \(G_S\) consists of a Stinespring dilation  \(E_{G_S}:\mathcal{H}_{p}^{\otimes(\dim(L)-n)}\cong \mathcal{H}_{G_S} \rightarrowtail \mathcal{H}_p^{\otimes n}\) of \(\rho_{G_S} \coloneq \Pi_{G_S}/\Tr(\Pi_{G_S})\), meaning that  \(E_{G_S} \circ \rho_{\mathsf{max}} = \rho_{G_S}\).
\(E_{G_S}\) encodes \(\dim(L)-n\) \textbf{logical qupits} into \(n\) \textbf{physical qupits}.
In the symplectic picture, we represent the encoder by the affine Lagrangian isometry \(\rel(E_{G_S}): (\Zp^{2(\dim(L)-n)}, \omega_{\dim(L)-n})\rightarrowtail (\Zp^n, \omega_n),\)
where: 
\begin{equation}
\im(\rel(E_{G_S})) =  \rel(E_{G_S} \circ \rho_{\mathsf{max}})= \rel(\rho_{G_S}) = \rel(\Pi_{G_S})= S.
\end{equation}

After encoding, \textbf{errors} may occur. Because the Pauli group \(\Pauli\) is a nice unitary error basis for \(\mathcal{H}_p^{\otimes n}\), we can interpret potential errors as linear combinations of Pauli operators \cite[\S~10.3.1.]{NC}.  In particular, this means that the error correction properties of the stabiliser code \(G_S\) can be captured either by reference to the isotropic subspace \(L^{\omega_n}\) or the projector \(\Pi_{G_S}\):
\begin{center}
\begin{tabular}{l|c|c}
\emph{An error \(\pi(a,\vb{e})\in \Pauli\) is} & \emph{Symplectic condition} & \emph{Projector condition} \\ \hline
\textbf{Trivial} 
& \(\vb{e}\in L^{\omega_n}\)
& \(\Pi_{G_S} \pi(0,\vb{e})\Pi_{G_S}=\Pi_{G_S}\) \\ \hdashline
\textbf{Detectable} 
& \(\vb{e}\notin L\)
& \(\Pi_{G_S} \pi(0,\vb{e}) \Pi_{G_S}=0\)\\  \hdashline
\textbf{Undetectable} 
& \(\vb{e}\in L\setminus L^{\omega_n}\)
& \(\Pi_{G_S} \pi(0,\vb{e}) \Pi_{G_S}\neq 0\),\\ 
\textbf{and nontrivial} && \(\Pi_{G_S} \pi(0,\vb{e})\Pi_{G_S}\neq \Pi_{G_S}\) \\
\end{tabular}
\end{center}
Errors that are undetectable and nontrivial are exactly the \emph{logical} errors that change the encoded data.
Moreover, a set \(\mathcal{E}\subseteq\F_p^{2n}\) of errors is
\textbf{correctable} if and only if
\begin{equation}
  \forall \vb{e}\neq\vb{f}\in\mathcal{E}: \vb{f}-\vb{e}\notin L \setminus L^\omega,
\end{equation}
which is to be contrasted with the more involved projector condition:
\begin{equation}
  \forall \vb{e}\neq\vb{f}\in\mathcal{E}: 
  \Pi_{G_S} \pi(\vb{0},\vb{f}-\vb{e}) \Pi_{G_S} = 0
  \qor 
  \Pi_{G_S} \pi(\vb{0},\vb{f}-\vb{e}) \Pi_{G_S} = \Pi_{G_S}.
\end{equation}
The \textbf{code distance} \(d(G_S) \in \N\) is the minimal number of tensor factors on which a nontrivial undetectable Pauli acts, and is considered a good measure of the quality of a code. This is most easily understood in the symplectic picture:
\begin{equation}
  d(S)
  \coloneqq 
  \min\left\{
  \ \big| \{ i\in\{0,\dots,n-1\} : (e_{x,i},e_{z,i})\neq(0,0) \} \big|
  \  : \
  \forall \vb{e}=(\vb{e}_x,\vb{e}_z)\in L\setminus L^{\omega_n}
  \ \right\}
\end{equation}

\paragraph*{Coarse-graining via mixing}

The inclusion order on nonempty affine Lagrangian subspaces, and their corresponding projectors is trivial. However, for affine coisotropic subspaces, it tells us when the corresponding projectors are more or less pure than each other: because \(\rel(\Pi_{G_S}) = S\) there is an inclusion of affine coisotropic subspaces \( R \subseteq S\) if and only if there is an inclusion of the images of their corresponding projectors \(\im(\Pi_{G_R }) \subseteq \im(\Pi_{G_S }) \). That is to say we have such an inclusion if and only if \(\Pi_{G_R } \preceq \Pi_{G_S}\)  in the  L\"owner order \cite[Section 6]{Arias2022}.

This order allows us to quantify and coarse-grain the error correction properties of stabiliser codes.
Given affine coisotropic subspaces \(R,S \subseteq (\mathbb{F}_p^{2n},\omega_n)\), if
\(R \subseteq S\) then  the constraints imposed by the stabiliser code  \(G_R\) can be relaxed to the constraints imposed by the stabiliser code \(G_S\), by forgetting some of the stabilisers. In particular, it follows straightforwardly that:
\begin{itemize}
\item a trivial error for \(G_S\) is also trivial for \(G_R\);
\item a detectable error for \(G_S\) is also detectable for \(G_R\);
\item a set of correctable errors for \(G_S\) is also correctable for \(G_R\);
\item \( d(G_S) \leq d(G_R) \). 
\end{itemize}
Of course, this comes with a trade-off, as the dimension of the coisotropic \(R\) decreases, so does the dimension of the logical space \(R \setminus R^{\omega_n}\).

This order on affine coisotropic subspaces gives a partial order enrichment of
\(\ACR\) and \(\Total(\ACR)\), i.e. a \emph{compositional} account for this
quantitative property of mixed stabiliser quantum \emph{processes}, where
encoders, stabiliser codes and Pauli errors can be composed together. That
is to say, we can ask when processes are coarse-grainings of each other:
\begin{theorem}
With respect to the L\"owner order on the corresponding projectors,
\begin{itemize}
    \item \(\CPM(\Stab_p)\) is enriched in preordered sets;
    \item \(\Proj(\CPM(\Stab_p))\) is enriched in partially ordered sets;
    \item  \(\Caus(\CPM(\Stab_p))\) is enriched in partially ordered sets.
\end{itemize}
\end{theorem}
\begin{proof}
 Since \(\ACR\) is compact closed, any morphism \(R:(V,\omega_V) \to (W, \omega_W)\) canonically induces a state \(\lfloor R \rfloor: I \to  (V\oplus W,\omega_V \oplus \omega_W) \), called its \emph{name}. For example in \(\CPM(\Stab_p)\) the name is proportional to the Choi matrix.
Therefore, there is an inclusion of affine coisotropic relations \( R \subseteq S\)
  if and only if there is an inclusion of the images of their corresponding projectors \(\im(\Pi_{G_{\lfloor R\rfloor} }) \subseteq \im(\Pi_{G_{\lfloor S\rfloor }}) \).

Just as dividing the projector onto the code space by the trace norm \(\Pi_{G_S}/ \Tr(\Pi_{G_S})\) is a mixed stabiliser state; dividing a  mixed stabiliser state by its \emph{operator norm} is a projector onto the code space of a stabiliser code.
Because the operator normalised Choi-matrices of morphisms between any two Hilbert spaces in \(\Caus(\CPM(\Stab_p))\) are distinct from each other, the equivalence \(\rel: \Caus(\CPM(\Stab_p))\simeq \Total(\ACR)\) induces a partial order on the trace preserving mixed stabiliser quantum processes.

On the other hand, when working in \(\CPM(\Stab_p)\), two proportional, yet distinct morphisms of the same type induce the same projector, so the order fails to be anti-symmetric.
\end{proof}
\section{Measurement and classical types}
\label{section:measurement}
In the previous section, we saw how the CPM construction provides an abstract setting for:
\begin{enumerate}
    \item general mixed state quantum mechanics when applied to \(\FHilb\);
    \item more specifically, stabiliser codes when applied to \(\ALR \simeq \Proj(\Stab_p)\).
\end{enumerate}
Stabiliser codes are used to detect and correct errors on noisy quantum channels: encoding quantum information redundantly as a mixed state. However, from, an operational point of view, to detect errors, one has to measure part of the code space; and to correct errors, one must apply operations to the code space conditional on the measurement outcomes. 

Given an indexed orthonormal basis \(B=\{\ket{\lambda_1}, \cdots,
\ket{\lambda_n}\}\) for a finite-dimensional Hilbert space \(\mathcal{H}\),
decoherence in this basis is represented by the completely positive trace-preserving map \({\mathcal E}_B:
\mathcal{B}(\mathcal{H})\to \mathcal{B}(\mathcal{H})\) which sends arbitrary
pure states to probabilistic mixtures of pure states in \(B\) according to
the \emph{Born rule}:
\begin{equation}\label{eq:collapse}
   {\mathcal E}_B(\dyad{\phi}{\phi}) \coloneqq
  \sum_{j=1}^n \dyad{\lambda_j}{\lambda_j}\dyad{\phi}{\phi} \dyad{\lambda_j}{\lambda_j}
  = \sum_{j=1}^n  \abs{\bra{\lambda_j}\ket{\phi}}^2\dyad{\lambda_j}{\lambda_j}.
\end{equation}
This operation already exhibits the essential structure of a measurement:
the diagonal states \(\dyad{\lambda_j}{\lambda_j}\) represent distinct
measurement outcomes, and probabilistic mixtures of them correspond to
 probability distributions over those outcomes. However, this
doesn't yet distinguish between classical and quantum systems. Not
all states in \(\mathcal{B}(\mathcal{H})\) are classical in this sense, nor will all quantum processes acting on \(\mathcal{B}(\mathcal{H})\)
preserve the classicality of states. Therefore, it is natural to ask how to enforce
this condition \emph{explicitly}, at the level of objects?

Observe that a quantum state \(\rho\) is a probabilistic mixture of pure states in the basis \(B\) if and
only if \(\mathcal{E}_B(\rho) = \rho\).  In other words, \(\mathcal{E}_B\) 
characterises the subalgebra of states in \(\mathcal{B}(\mathcal{H})\)
that are classical with respect to \(B\); i.e. the states that are diagonal in the basis \(B\). We can impose a classical
structure on a quantum system by restricting to such states, and similarly for general processes. 
In other words, \(\mathcal{E}_B\) captures an internal classical structure of a quantum object,
making it possible to construct a new category where these decoherence maps themselves become objects, and morphisms between them respect this classical structure. 

\subsection{Adding classical types by splitting dagger-idempotents}
The following structure will serve as our abstract notion of a quantum-classical system in a \dag-CCC:
\begin{definition}
    A \textbf{\dag-idempotent} in a \dag-SMC is a map \(f\) such that \(f^\dagger = f\) and \(f;f = f\).
\end{definition}
In the setting of finite-dimensional, mixed quantum theory:
\begin{example}[{\cite[Thm. 2.5]{heunen_completely_2014}, \cite[Prop. 3.5]{Coecke2014}}]
    The \dag-idempotents on \(\mathcal H\) in \(\CPM(\FHilb)\) are in bijection with \(C^*\)-subalgebras of the matrix algebra \(\mathcal{B}(\mathcal{H})\).
\end{example}
The identity on \(\mathcal{H}\) is a \dag-idempotent, corresponding to the matrix algebra \(\mathcal{B}(\mathcal{H})\). 
On the other hand, decoherence via a basis measurement \(\mathcal{E}_{B}\) is a \dag-idempotent corresponding to a commutative \(C^*\)-subalgebra of \(\mathcal{B}(\mathcal{H})\).  The following definition allows us to promote these \(C^*\)-algebras into objects in their own right:
\begin{definition}[{\cite[Def. 3.13]{selinger_idempotents_2008}}]
  \label{def:split}
  Given a \dag-CCC \(\mathsf{C}\), its \dag-idempotent completion \(\Split(\mathsf{C})\) is the \dag-CCC with:
  \begin{itemize}[nosep]
    \item \textbf{objects:}  pairs \((A,a)\) where \(A\) is an object of
    \(\mathsf{C}\) and \(a : A \to A\) is a \dag-idempotent;
    \item \textbf{morphisms:} \(f:(A,a) \to (B,b)\) are morphisms \(f :
    A \to B\) in \(\mathsf{C}\) such that \(a;f;b = f\);
    \item \textbf{identities:} \(1_{(A,a)}\coloneqq a \);
    \item \textbf{rest of \dag-compact structure}:
    given pointwise in \(\mathsf{C}\).
  \end{itemize}
\end{definition}

There is a canonical embedding \(\mathsf{C}\rightarrowtail \Split(\mathsf{C})\) sending objects \(A \mapsto (A, 1_A)\) and acting as the identity on morphisms.  
When applied to \(\CPM(\FHilb)\), the \dag-idempotent completion reproduces the standard setting for finite-dimensional quantum mechanics; augmenting the category of finite dimensional matrix algebras and completely positive maps so that the objects become arbitrary finite-dimensional \(C^*\)-algebras:

\begin{theorem}[{\cite[Thm. 2.5]{heunen_completely_2014}, \cite[Prop. 3.5]{Coecke2014}}]
    \(\Split(\CPM(\FHilb))\) is equivalent to the CCC of completely-positive maps between finite-dimensional \(C^*\)-algebras.
\end{theorem}

The objects of the form \((\mathcal{H},1_{\mathcal{H}})\) represent the matrix algebras \(\mathcal{B}(\mathcal{H})\), interpreted as the purely quantum systems in \(\CPM(\FHilb)\).
On the other hand, the new objects added by \dag-idempotent completion correspond to other \(C^*\)-algebras. 
For the extremal example of a quantum measurement induced by an orthonormal basis \(B\), the object \((\mathcal{H},\mathcal{E}_B)\) is interpreted as the state space of a classical system measured according to the basis \(B\). The canonical map \(\mathcal{E}_B:(\mathcal{H},1_{\mathcal{H}}) \twoheadrightarrow(\mathcal{H},\mathcal{E}_B)\) is interpreted as the \textbf{measurement} induced by \(B\); whereas the map \(\mathcal{E}_B:(\mathcal{H},\mathcal{E}_B)\rightarrowtail(\mathcal{H},1_{\mathcal{H}}) \) is interpreted as the \textbf{classically-conditioned state preparation} in the basis \(B\).
Measurement followed by state preparation yields the quantum system projected onto the measurement basis; whereas, state preparation followed by measurement yields the identity on the classical system:
\begin{center}\def\arraystretch{.8}
\begin{tabular}{rlcrl}
\begin{tabular}{c} \em Measuring\\ \em then\\ \em preparing:\end{tabular}&
\(\begin{tikzcd}[row sep=.23cm]
	{(\mathcal{H}, 1_{\mathcal{H}})} \\
	{(\mathcal{H}, \mathcal{E}_B)} & {(\mathcal{H}, 1_{\mathcal{H}})}
	\arrow["{\mathcal{E}_B\ }"', tail, from=1-1, to=2-1]
	\arrow["{\mathcal{E}_B}", from=1-1, to=2-2]
	\arrow["{\mathcal{E}_B}"', two heads, from=2-1, to=2-2]
\end{tikzcd}\)& &
\begin{tabular}{c} \em Preparing\\ \em then\\ \em measuring:\end{tabular}&
\(\begin{tikzcd}[row sep=.23cm]
	{(\mathcal{H}, \mathcal{E}_B)} & {(\mathcal{H}, 1_{\mathcal{H}})} \\
	& {(\mathcal{H}, \mathcal{E}_B)}
	\arrow["{\mathcal{E}_B}", two heads, from=1-1, to=1-2]
	\arrow["{\mathcal{E}_B}"', equal, from=1-1, to=2-2]
	\arrow["{\ \mathcal{E}_B}", tail, from=1-2, to=2-2]
\end{tikzcd}\)
\end{tabular}
\end{center}
This makes explicit the fact that \(\mathcal{E}_B\) acts as the decoherence map on the quantum object \((\mathcal{H},1_\mathcal{H})\) but as the identity on the classical object \((\mathcal{H},\mathcal{E}_B)\).

Arbitrary completely-positive maps between \(C^*\)-algebras cannot be physically implemented.  Just as in the previous section, we must impose an additional normalisation constraint:
\vspace{-5mm}
\begin{definition}
    A morphism \([f,S] : (X,x) \to (Y,y)\) in \(\Split(\CPM(\mathsf{C}))\) is \textbf{causal} if and only if \(\Tr_Y[f,S] = \Tr_X :(X,x) \to (I, 1_I) \).
     Denote  the symmetric monoidal subcategory of causal morphisms in \(\Split(\CPM(\mathsf{C}))\) by \(\Caus(\Split(\CPM(\mathsf{C})))\).
\end{definition}

In the setting of finite-dimensional quantum theory; this reproduces the usual operator algebraic setting for finite-dimensional quantum mechanics:
\begin{corollary}
\(\Caus(\Split(\CPM(\mathsf{\FHilb})))\) is equivalent to the symmetric monoidal category of completely-positive trace-preserving (CPTP) maps between finite-dimensional \(C^*\)-algebras.
\end{corollary}

In other words, the morphisms in \(\Caus(\Split(\CPM(\FHilb)))\) correspond to finite\hyp{}dimensional \textbf{quantum channels}, and the states correspond to \textbf{density matrices}.
Importantly, the state preparation and measurement maps are quantum channels.
For notational convenience, denote the category of quantum channels by
\(\QChan \coloneqq \Caus(\Split(\CPM(\FHilb))).\)

\subsection{The stabiliser theory with affine classical operations}

In the previous subsection, we reviewed how the symmetric monoidal category \(\QChan\) is equivalent to the standard setting for finite-dimensional quantum mechanics with measurement and classical control. In this subsection, by applying the same constructions to the symplectic representation of mixed stabiliser maps; we show that the canonical setting for stabiliser quantum mechanics with Pauli measurement and Pauli state preparation admits a concise, entirely relational description. To this end, we recall the following definition which will serve as the setting for our classical systems:
\begin{definition}
The \dag-compact-closed category, \(\AR\), of \textbf{affine relations} has finite-dimensional \(\Zp\)-vector spaces as objects and affine subspaces as morphisms.  Composition is given by relational composition, whilst the identity and compact-closed structure are given by the diagonal relations.  The dagger is given by the relational converse.
\end{definition}
\begin{lemma}
    There is a faithful \dag-compact-closed functor \(Q:\ALR\rightarrowtail \AR\) which forgets symplectic structure.
\end{lemma}

Instead of directly forming \(\Split(\ACR)\) to add classical state spaces to \(\ACR\), we can add additional affine relations to the image of \(Q:\ACR\rightarrowtail \AR\) which \dag-split \dag-idempotents:
\begin{proposition}
    The \dag-idempotents in \(\ACR\) \dag-split through the forgetful functor \(\ACR\rightarrowtail \AR\).
    In particular, Pauli-\(Z\) measurement splits through the relations:
     \[\mu_Z\!\coloneqq\!\left\{ \left( \bigl[\begin{smallmatrix} x\\z \end{smallmatrix}\bigr], x\right) \!\in\! \Zp^2\oplus\Zp\right\}\!:\! Q(\Zp^2, \omega_2)\!\twoheadrightarrow\!\Zp,\ \ \text{and}\ \ \mu_Z^\dag \!\coloneqq\!\left\{\left( x,\bigl[\begin{smallmatrix} x\\z \end{smallmatrix}\bigr]\right)\!\in\! \Zp\oplus\Zp^2\right\}\!:\! \Zp\!\rightarrowtail\!Q(\Zp^2, \omega_2).\]
\end{proposition}
\begin{proof}
    Consider the relation in \(\ACR\) corresponding to the \(Z\)-decohering channel:
    \[\mathscr{E}_Z \coloneqq\rel(\mathcal{E}_Z)= \left\{\left( \bigl[\begin{smallmatrix} x\\z \end{smallmatrix}\bigr], \bigl[\begin{smallmatrix} x\\z' \end{smallmatrix}\bigr]\right)\in\Zp^2\oplus\Zp^2 \right\}:(\Zp^2, \omega_2) \to (\Zp^2, \omega_2)\]
    Any \dag-idempotent in \(\ACR\) is symplectomorphic to \(\mathscr{E}_Z^{\oplus n} \oplus 1_{(\Zp^{2m},\omega_m)}\) for some \(n,m \in \N\). Moreover, \(Q(\mathscr{E}_Z)= \mu_Z^\dag\circ \mu_Z\) splits as \(\mu_Z\circ \mu_Z^\dag = 1_{\Zp}\).
\end{proof}
In other words whereas the state space of a qupit is represented by the symplectic vector space \((\F_p^2,\omega_1)\), we only need the 1-dimensional vector space \(\F_p\) to represent a classical state space.
The process of \dag-splitting  \dag-idempotents through \(Q:\ACR\rightarrowtail \AR\) adds a \emph{quantum modality} \(Q\) to \(\AR\); imposing compatibility with the symplectic structure:
\begin{definition}
    Let \(\ARQ\) denote the \dag-CCC with:
    \begin{itemize}[nosep]
        \item \textbf{Objects:} generated by finite direct sums of:
        \begin{itemize}
            \item finite dimensional symplectic vector spaces \(Q(V, \omega_V) \in Q(\ACR)\); 
            \item  finite dimensional vector spaces \(W \in \AR\);
        \end{itemize}
        \item \textbf{Morphisms:} generated by \(Q(\ACR)\) in addition to the morphisms under the direct sum and composition of:

        \begin{itemize}
            \item the coisometry  \(\mu_Z: Q(\Zp^2, \omega_2) \twoheadrightarrow \Zp \), corresponding to measurement in the \(Z\)-basis;
            \item   the isometry \(\mu_Z^\dag: \Zp\rightarrowtail Q(\Zp^2, \omega_2) \), corresponding to state preparation in the \(Z\)-basis.
        \end{itemize}
    \end{itemize}
\end{definition}
That is to say that if we have an object \(Q(\F_p^{2n}, \omega_n)\), we know it represents an n-qupit quantum system; whereas an object \(\F_p^n\) represents an n-pit classical system.
By restricting to either class of objects, this justifies our interpretation of \(Q(\F_p^{2n}, \omega_n)\) as the quantum system and \(\F_p^n\) as a classical system:
\begin{lemma}
    \(\ACR\) and \(\AR\) are full \dag-compact closed subcategories of \(\ARQ\).
\end{lemma}

Moreover, because \(Q:\ACR\rightarrowtail \AR\) is faithful, adding these two relations produces the same thing as directly adding the classical objects through the process given in  definition~\ref{def:split}:
\begin{theorem}
    \label{theorem:splitting_modality}
    There is a \dag-compact closed equivalence \(\Split(\ACR) \simeq \ARQ\).
\end{theorem}

In other words, this category is obtained by glueing together the quantum \dag-CCC \(\ALR\) with the classical \dag-CCC \(\AR\) along the measurement map \(\mu_Z\) and the state preparation map \(\mu_Z^\dag\).
Finally, now that we can capture quantum-classical interfaces using the symplectic representation, our moto becomes:
\begin{center}
  \boxed{\emph{Everything is an affine relation, with quantum data captured by a \underline{symplectic modality}!}}
\end{center}
...which, admittedly, is not \emph{quite} as catchy as the previous motos.\\

However, just as for \(\Split(\CPM(\FHilb))\); the category \(\Proj(\Split(\CPM(\Stab_p))\simeq \ARQ\) has morphisms which do not correspond to operations which can be physically implemented.  We restrict ourselves to the completely-positive trace-preserving maps:

\begin{proposition}
  \label{proposition:causal_total_split}
  The induced functor \(\rel:\Split(\CPM(\Stab_p)) \to \ARQ\) restricts to a symmetric monoidal equivalence \(\Caus(\Split(\CPM(\Stab_p)))\simeq \Total(\ARQ)\) making the following diagram of symmetric monoidal categories commute:
\[\begin{tikzcd}[row sep=.25cm, cells={nodes={inner sep=1pt, outer sep=0pt}}]
	{\Total(\ARQ)} && {\ARQ} \\[-.1cm]
	{\Caus(\Split(\CPM(\Stab_p)))} & {\Split(\CPM(\Stab_p))} & {\Proj(\Split(\CPM(\Stab_p)))} \\
    {\Caus(\Split(\CPM(\FHilb)))} & {\Split(\CPM(\FHilb))} & {\Proj(\Split(\CPM(\FHilb)))} \\
	\begin{array}{c} \begin{matrix}\text{CPTP maps between}\\[-2pt]\text{f.d.\ \(C^*\)-algebras}\end{matrix} \end{array} & \begin{array}{c} \begin{matrix}\text{CP maps between}\\[-2pt]\text{f.d.\ \(C^*\)-algebras}\end{matrix} \end{array} & \begin{array}{c} \Proj\left(\begin{matrix}\text{CP maps between}\\[-2pt]\text{f.d.\ \(C^*\)-algebras}\end{matrix}\right) \end{array}
	\arrow[tail, from=1-1, to=1-3]
	\arrow["\!\simeq"{marking, allow upside down}, draw=none, from=1-1, to=2-1]
	\arrow["\!\simeq"{marking, allow upside down}, draw=none, from=1-3, to=2-3]
	\arrow[tail, from=2-1, to=2-2]
	\arrow[tail, from=2-1, to=3-1]
	\arrow[two heads, from=2-2, to=2-3]
	\arrow[tail, from=2-2, to=3-2]
	\arrow[tail, from=2-3, to=3-3]
	\arrow[tail, from=3-1, to=3-2]
	\arrow[two heads, from=3-2, to=3-3]
    \arrow["\!\simeq"{marking, allow upside down}, draw=none, from=3-1, to=4-1]
    \arrow["\!\simeq"{marking, allow upside down}, draw=none, from=3-2, to=4-2]
    \arrow["\!\simeq"{marking, allow upside down}, draw=none, from=3-3, to=4-3]
    \arrow[tail, from=4-1, to=4-2]
    \arrow[two heads, from=4-2, to=4-3]
\end{tikzcd}\]

\end{proposition}
\begin{proof}
     This follows from essentially the same argument as for proposition~\ref{proposition:causal_total_cpm}.
\end{proof}

Note that classically-controlled Pauli operators can be represented in \(\Total(\ARQ)\) because they can be constructed with Clifford operators as well as Pauli state preparation and measurements.  For notational convenience, from here onwards, denote the subcategory of \(\QChan\) whose morphisms are qupit \textbf{stabiliser quantum channels} by 
\[\StabChan_p \coloneqq \Caus(\Split(\CPM(\Stab_p))).\]

Because all of the morphisms in \(\Total(\ARQ)\) are affine subspaces of vector spaces over \(\Zp\), this means that using the representation \(\Total(\ARQ)\simeq \StabChan_p\), the  exact equality of stabiliser quantum channels is computable in  deterministic polynomial time. This is a deterministic analogue of the celebrated, probabilistic Gottesman-Knill theorem \cite{aaronson2004improved}.

Again, note that everything carries through for qubits when restricting to the setting of CSS maps, except that the relations for the Pauli-\(X\) measurement and state preparation must also be included, because the Hadamard gate is not a CSS Clifford.
\subsection{Arbitrary classical operations}
\label{measurement:nonlinear_control}

By relaxing the requirement that the relations between classical objects are affine relations we obtain the following category:
\begin{definition}
    Let \(\NLQ\) denote the \dag-compact closed category given by the objects and morphisms of \(\ARQ\) in addition to the \emph{non-affine} relation between classical types:
    \[ \left\{ \left( \bigl[\begin{smallmatrix} a \\ b \end{smallmatrix}\bigr], a\cdot b \right) \in \Zp^2 \oplus \Zp\right\}:\Zp^2 \to \Zp.\]
\end{definition}
Indeed, since all functions \(\F_p^n \to \F_p^m\) can be written as polynomials, it follows that:
\begin{lemma}
The morphisms from \(\Zp^n\) to \(\Zp^m\) in \(\NLQ\) are precisely set-relations from the set \(\Zp^n\) to the set \(\Zp^m\) , ie. subsets of \(\Zp^n \oplus \Zp^m\).
\end{lemma}

Therefore, the total relations between classical objects are precisely the total relations between the underlying sets.  These are exactly the classical corrections which can be performed nondeterministically.
However, by relaxing the affine constraints between classical objects, the morphisms between quantum objects in \(\NLQ\) also fail in general to be affine coisotropic subspaces, or even just affine subspaces at all.  In particular, this means that the symplectic representation of stabiliser circuits fails:

\begin{proposition}
  There is no functor \(\Total(\NLQ)\to \StabChan_p\) making the following diagram commute:
  \begin{center}
  \(
    \begin{tikzcd}[column sep=32pt, row sep=10pt, cells={nodes={inner sep=1.5pt}}]
      \Total(\ARQ) \arrow[tail]{r} \arrow[dr, "\simeq"'] &
      \Total(\NLQ)
        \arrow[d, densely dotted,
               start anchor=south, end anchor=north,
               shorten <=0pt, shorten >=0pt,
               "\nexists"{yshift=2pt}] \\
      &
      \StabChan_p
    \end{tikzcd}
  \)
  \end{center}
\end{proposition}

This is because the non-affine classical control of stabiliser codes can produce mixed states which are no longer proportional to uniform mixtures of pure states;   non-affine corrections between basis elements can create mixtures of pure states with non-uniform weights. In other words,  \(\Total(\NLQ)\) can not tell us the probability of measurement outcomes.
\section{Case study: a small imperative language for stabiliser QEC}
\label{section:language}
In this section, we introduce a minimal imperative language \(\StabLang\) (\emph{Stabiliser Programming Language}) for stabiliser quantum channels. In other words, this is a language for quantum error
correction, including measurements and classical control. This language
is strongly inspired by the language \(\QPL\) \cite{selinger_towards_2004}; where now our quantum operations are 
restricted to stabiliser operations and our  classical operations are restricted to total  \emph{affine} relations.  We have implemented SPL and its denotational semantics in Python \cite{SPL}.

We give \(\StabLang\) small‐step operational semantics on pairs \([C|\rho]\) of terms
acting on density operators as CPTP maps (similar to that of Ying \cite[Section
3.2]{mingsheng_foundations_2016}), and a fully abstract denotational semantics in the SMC
\(\Total(\AR^Q)\). This case study serves
as a proof-of-concept to demonstrate that our symplectic semantics can be
used as the foundation of a quantum compilation stack whose target code is
fault-tolerant by construction and with a denotational semantics amenable to
formal verification. The purpose of \(\StabLang\) is to show that \(\Total(\AR^Q)\) serves as a denotational semantics for stabiliser quantum programs, and is to be contrasted with more powerful, and computationally expressive languages such as Quipper \cite{quipper} and Proto-Quipper \cite{protoquipper} which are not specifically tailored to the stabiliser fragment.

\subsection{Syntax}

\begin{figure}
  \begin{mathpar}
    \inferrule {\Gamma \vdash c \update \Delta \quad \Delta \vdash d \update \Sigma}
    {\Gamma \vdash c\;\fby \; d \update \Sigma}
    \qquad
    \inferrule {\quad} {\Gamma \vdash \pit \; \reg{x} \update \reg{x} : \pittype, \Gamma} 
    \qquad
    \inferrule {\quad} {\Gamma \vdash \qpit \; \reg{x} \update \reg{x} : \qpittype, \Gamma} 
  \end{mathpar}
  \vspace{-6mm}
  \begin{mathpar}
    \inferrule {\quad} {\vbreg{x} : \pittype^n, \Gamma \vdash \vbreg{y}\!=\!A\affmul \vbreg{x} \update \vbreg{x} : \pittype^n, \vbreg{y} : \pittype^m, \Gamma}
    \qquad
    \inferrule {\quad} {\reg{x} : \qpittype, \Gamma \vdash \measure \; \reg{x} \update \reg{x} : \pittype, \Gamma}   
  \end{mathpar}
  \vspace{-5mm}
  \begin{mathpar}
    \inferrule {\quad} {\Gamma \vdash \noop \update \Gamma} 
    \qquad
    \inferrule {\quad} {\reg{x}: \pittype, \Gamma \vdash \discard \; \reg{x}\update \Gamma}
    \qquad
    \inferrule {\quad} {\vbreg{x} : \qpittype^n, \Gamma \vdash \vbreg{x} \muleqq U  \update \vbreg{x}: \qpittype^n, \Gamma}
  \end{mathpar}
  \vspace{-5mm}
  \begin{mathpar}
    \inferrule {\quad}
    {\reg{x} : \pittype, \reg{y} : \qpittype, \Gamma \vdash \control_P \; \reg{x} \; \reg{y} \update \reg{x} : \pittype, \reg{y} : \qpittype, \Gamma}
  \end{mathpar}
  \caption{\textsf{\textbf{Formation rules for \(\StabLang\)}}. New variables are always assumed to be fresh.\\
  \(n\in \N^{>0}\), \(\tau\in \Ty\), and  \(\vbreg{x}:\tau^n\) is shorthand for \(\{\reg{x1}:\tau,\cdots, \reg{xn}:\tau\}\) such that  \(\vbreg{x} = (\reg{x1},\cdots, \reg{xn}) \in \Reg^n\).}
  \label{language:fig:formation_rules}
\end{figure}

\(\StabLang\) has classical and quantum types \( \Ty  \;\Coloneqq\; \pittype
\;\mid\; \qpittype \). The terms are generated from the following grammar
with respect to some fixed, linearly ordered set \(\Reg\) indexing registers:
  \[
  \begin{aligned}
  c,d \;\Coloneqq\; 
     c\fby d
    \mid \pit \; \reg{x}
    \mid \vbreg{y}\!=\!A\affmul \vbreg{x}
    \mid \discard \;\reg{x}
    \mid\qpit \; \reg{x}
    \mid \vbreg{x}  \muleqq U
    \mid \measure \;\reg{x}
    \mid \control_P \;\reg{x}\;\reg{y}\;
    \mid \noop.
  \end{aligned}
  \]
for all \(n,m\in\N\), \(U \in \Cliff, P \in \Pauli\), \(\Zp\)-affine
transformations \(A:\Zp^n \to \Zp^m\),  and \(\reg{x},\reg{y}\in \Reg\),
\(\vbreg{x}\in \Reg^n\), \(\vbreg{y} \in \Reg^m\).
We interpret each term as a quantum operation:
\begin{itemize}
    \item \(c\fby d\) represents the sequential composition of subterms;

    \item \(\pit \; \reg{x}\) represents the initialisation of \(\reg{x}\) as the
pit (\(p\)-ary digit) \(0\);

    \item \(\vbreg{y}\!=\!A\affmul \vbreg{x}\) applies the
affine transformation \(A\) to \(\vbreg{x}\) and stores the result on
\(\vbreg{y}\);

    \item \(\discard \;\reg{x}\) takes the trace of \(\reg{x}\);

    \item \(\qpit \; \reg{x}\) represents initialisation of \(\reg{x}\) as the qupit
\(\ket{0}\);

    \item \(\vbreg{x}  \muleqq U\) applies the Clifford operator \(U\)
on \(\vbreg{x}\);

    \item \(\measure \;\reg{x}\) represents the Pauli-\(Z\)
measurement on \(\reg{x}\);

    \item \(\control_P\;\reg{x}\;\reg{y}\;\) applies the
Pauli \(P\) operator on \(\reg{y}\), classically controlled by \(\reg{x}\) in the \(Z\)-basis;

    \item \(\noop\) represents the identity operation.
\end{itemize}

\(\StabLang\) is equipped with an \emph{environment-transforming} type system, which enforces linear usage of quantum data.
Typed environments are partial functions \(\Gamma:\Reg\to \Ty\) which bind registers to be either qupits or pits, and which we will represent by \(\{\reg{x} : \tau,\; \reg{y} : \sigma,\; \reg{z} :\mu,\; \cdots\}\) for \(\reg{x},\reg{y},\reg{z}\cdots \in \Reg\) and \(\tau,\sigma,\mu,\cdots \in \Ty\). We impose that the domain \(\dom(\Gamma)\) of \(\Gamma\), i.e. the set of bound registers \(\{\reg{x},\reg{y},\reg{z},\cdots \}\), is \emph{finite}, and that every register has at most one type. Judgments are triples \(\Gamma \vdash t \update \Delta\) consisting of a term \(t\) and typed environments \(\Gamma,\Delta\). A judgment \(\Gamma \vdash t \update \Delta\) is \textbf{well-formed} if it is derivable from the formation rules given in figure~\ref{language:fig:formation_rules}.

\begin{figure}
    \centering
    \[
        \begin{aligned}
            &\{ \reg{in} : \qpittype \} \vdash\ && \\
            &\qpit\ \reg{x}\ \fby\ \qpit\ \reg{out} \ \fby
            &&\texttt{\% initialize registers} \\
            & \reg{x}\ \muleqq\ F\ \fby\ (\reg{x},\reg{out})\ \muleqq\ C^X \ \fby 
            &&\texttt{\% prepare Bell pair} \\
            &  (\reg{in},\reg{x})\ \muleqq\ C^X\ \fby\ \reg{in}\ \muleqq\ F \ \fby\ \measure\ \reg{in}\ \fby\ \measure\ \reg{x}  \ \fby
            &&\texttt{\% Bell measurement} \\
            &  \control_{Z}\ \reg{in}\ \reg{out}\ \fby\ \control_{X}\ \reg{x}\ \reg{out}  \ \fby
            &&\texttt{\% phase correction} \\
            & \discard\ \reg{in}\ \fby\ \discard\ \reg{x} 
            &&\texttt{\% discard ancillae}\\
            & \!\!\update \{\reg{out} : \qpittype \} 
        \end{aligned}
    \]
    \caption{\textsf{\textbf{Qupit teleportation in \(\StabLang\)}}. 
    Where \(F\) is the Fourier transform; \(C^X\) is the (coherently) controlled \(X\) gate; and \(Z\) and \(X\) are the Pauli \(Z\) and \(X\) gates. Input is given on register  \(\reg{in} : \qpittype\) and output is returned on register  \(\reg{out} : \qpittype\).}
    \label{fig:stablang-teleport}
\end{figure}

\subsection{Operational semantics}
In this subsection we define a structured operational semantics for \(\StabLang\), which is strongly
inspired by Ying's operational semantics for quantum programs \cite[Section 3.2]{mingsheng_foundations_2016}.


We interpret typed environments as objects in the embedding \(\StabChan_p\rightarrowtail \QChan\):
\begin{definition}
     On types, let
    \(\langle \pittype \rangle \coloneqq (\mathcal{H}_p,\mathcal{E}_{Z}) \) and \(\langle \qpittype \rangle \coloneqq (\mathcal{H}_p, 1_{\mathcal{H}_p})\). Define the denotation of a typed environment to be the dependent tensor product \(\langle \Gamma\rangle \coloneqq \bigotimes_{\reg{x} \in \dom(\Gamma)}
    \langle\Gamma(\reg{x})\rangle\). 
    Moreover, let \(\mathcal{D}(\Gamma) \coloneqq \QChan(\C, \langle\Gamma\rangle )\) denote the set of density operators on \(\langle \Gamma \rangle\).
\end{definition}

To give our operational semantics, we establish notation to represent stabiliser quantum channels acting on subspaces of a larger ambient space.
Take typed environments \(\Gamma, \Delta\) and ordered subsets (ie. lists) \(\vbreg{x} \subseteq \dom(\Gamma)\) and \(\vbreg{y} \subseteq \dom(\Delta)\), where moreover, \(\dom(\Gamma) \setminus \vbreg{x} = \dom(\Delta) \setminus \vbreg{y}\).
Given a stabiliser quantum channel \(\mathcal{C}:\langle\Gamma|_{\vbreg{x}}\rangle\to \langle\Delta|_{\vbreg{y}}\rangle\) let \(\mathcal{C}_{\vbreg{x}, \vbreg{y}}:\langle \Gamma\rangle \to \langle \Gamma' \rangle\) be the stabiliser quantum channel acting as \(\mathcal{C}\) on the subspace \(\langle\Delta\rangle \subseteq \langle\Gamma\rangle\) and trivially on its orthogonal complement \(\langle\Gamma \setminus \Delta\rangle\subseteq \langle \Gamma\rangle\).
\begin{definition}
   A \textbf{configuration} is a pair  consisting of a well-formed judgement \(\Gamma \vdash t \update \Delta\) and a density operator \(\rho \in \mathcal{D}(\Gamma)\), denoted \(\Conf{\Gamma \vdash t \update \Delta}{\rho\in\mathcal{D}(\Gamma)}\), or \(\Conf{t}{\rho}\) for short.
    
    The \textbf{small-step operational semantics} of \(\StabLang\) is given by the reduction rules \(\,\rightsquigarrow\) in figure~\ref{fig:smallstep}. The operational semantics \(\,\rightsquigarrow^*\) is given by the transitive closure of \(\,\rightsquigarrow\).

\end{definition}

\begin{figure}[t]
  \begin{align*}
  \aden{\noop}{\Gamma}{\Gamma}
    &\coloneqq 1_{\langle \Gamma \rangle} : \langle \Gamma\rangle \to \langle \Gamma\rangle
  \\
  \aden{\pit\ \reg{x}}{\Gamma}{\reg{x}:\pittype,\Gamma}
    &\coloneqq \iota(\ket{0})_{\varnothing,\reg{x}}
    : \langle \Gamma\rangle \to \langle \reg{x}:\pittype,\Gamma\rangle
  \\
  \aden{\qpit\ \reg{x}}{\Gamma}{\reg{x}:\qpittype,\Gamma}
    &\coloneqq \iota(\ket{0})_{\varnothing,\reg{x}}
    : \langle \Gamma\rangle \to \langle \reg{x}:\qpittype,\Gamma\rangle
  \\
  \aden{\measure\ \reg{x}}{\reg{x}:\qpittype,\Gamma}{\reg{x}:\pittype,\Gamma}
    &\coloneqq (\mathcal{E}_Z)_{\reg{x},\reg{x}}
    : \langle \reg{x}:\qpittype,\Gamma\rangle \to \langle \reg{x}:\pittype,\Gamma\rangle
  \\
  \aden{\discard\ \reg{x}}{\reg{x}:\pittype,\Gamma}{\Gamma}
    &\coloneqq (\Tr_{\mathcal{H}_p})_{\reg{x},\varnothing}
    : \langle \reg{x}:\pittype,\Gamma\rangle \to \langle \Gamma\rangle
  \\
  \aden{\vbreg{y}\!=\!M\affmul\vbreg{x}}{\vbreg{x}:\pittype^{m},\Gamma}{\vbreg{x}:\pittype^{m},\vbreg{y}:\pittype^{n},\Gamma}
    &\coloneqq \mathcal{A}^{M}_{\vbreg{x},\vbreg{y}}
    : \langle \vbreg{x}:\pittype^{m},\Gamma\rangle \to \langle \vbreg{x}:\pittype^{m},\vbreg{y}:\pittype^{n},\Gamma\rangle
  \\
  \aden{\vbreg{x}\muleqq U}{\vbreg{x}:\qpittype^{n},\Gamma}{\vbreg{x}:\qpittype^{n},\Gamma}
    &\coloneqq \iota(U)_{\vbreg{x},\vbreg{x}}
    : \langle \vbreg{x}:\qpittype^{n},\Gamma\rangle \to \langle \vbreg{x}:\qpittype^{n},\Gamma\rangle
  \\
  \aden{\control_{P}\ \reg{x}\ \reg{y}}{\reg{x}:\pittype,\reg{y}:\qpittype,\Gamma}{\reg{x}:\pittype,\reg{y}:\qpittype,\Gamma}
    &\coloneqq C_P{}_{(\reg{x},\reg{y}),(\reg{x},\reg{y})}
    : \langle \reg{x}:\pittype,\reg{y}:\qpittype,\Gamma\rangle \to \langle \reg{x}:\pittype,\reg{y}:\qpittype,\Gamma\rangle
  \end{align*}
  \caption{\textsf{\textbf{Denotations of the atomic terms of \(\StabLang\) in \(\QChan\).}}
  Where 
  \(\iota:\Stab_p\to \QChan\) embeds pure stabiliser maps into stabiliser channels.
  \(\mathcal{E}_Z:(\mathcal{H}_p,1_{\mathcal{H}_p})\to(\mathcal{H}_p,\mathcal{E}_Z)\) is Pauli-\(Z\) measurement.
  \(\Tr_{\mathcal{H}_p}:(\mathcal{H}_p,\mathcal{E}_Z)\to I\) is trace.
  \(\mathcal{A}^{M}=\sum_{\vb{x}\in\Zp^{m}}\dyad{M(\vb{x})}{\vb{x}}:(\mathcal{H}_p,\mathcal{E}_Z)^{\otimes m}\to(\mathcal{H}_p,\mathcal{E}_Z)^{\otimes n}\) is the classical channel for an \(\F_p\)-affine map \(M:\F_p^{m}\to \F_p^{n}\).
  \(C_P:(\mathcal{H}_p,\mathcal{E}_Z)\otimes(\mathcal{H}_p,1_{\mathcal{H}_p})\to(\mathcal{H}_p,\mathcal{E}_Z)\otimes(\mathcal{H}_p,1_{\mathcal{H}_p})\) is classically controlled Pauli for \(P\in\mathcal{P}_p\).
  }
  \label{fig:aden-generators}
\end{figure}

\begin{figure}[t]
\begin{mathpar}
\inferrule{
  \Gamma \vdash a \update \Delta \\
  \Conf{\,\Gamma \vdash c \update \Delta\,}{\rho \in \mathcal{D}(\Gamma)}
   \rightsquigarrow
   \Conf{\,\Theta \vdash c' \update \Delta\,}{\rho' \in \mathcal{D}(\Theta)}
}{
  \Conf{\,\Gamma \vdash a \fby c \update \Sigma\,}{\rho \in \mathcal{D}(\Gamma)}
  \rightsquigarrow
  \Conf{\,\Delta \vdash c \update \Sigma\,}{\ \aden{a}{\Gamma}{\Delta} \circ \rho \ \in \mathcal{D}(\Delta)}
}

\inferrule{
  \Gamma \vdash a \update \Delta
}{
  \Conf{\,\Gamma \vdash a \update \Delta\,}{\rho \in \mathcal{D}(\Gamma)}
  \rightsquigarrow
  \Conf{\,\Delta \vdash \noop \update \Delta\,}{\ \aden{a}{\Gamma}{\Delta} \circ \rho \ \in \mathcal{D}(\Delta)}
}

\end{mathpar}
\caption{
Small-step operational semantics of \(\StabLang\), where \(a\) is an atomic term and \(c,c'\) are well-formed terms.
The interpretation of atomic terms \(a\) in \(\QChan\), denoted \(\aden{a}{\Gamma}{\Delta}:\langle \Gamma\rangle\to\langle \Delta\rangle\) is defined in figure~\ref{fig:aden-generators}.
}
\label{fig:smallstep}
\end{figure}

\begin{definition}
  Well-formed \(\StabLang\) terms \(\Sigma \vdash c \update \Delta,\Sigma
  \vdash d \update \Delta\) are \textbf{observationally equivalent} if for
  any well-formed terms \(\Gamma \vdash a \update \Sigma\) and \(\Delta \vdash
  b \update \Theta\), and density operator \(\rho \in \mathcal{D}(\Gamma)\),
   \(\Conf{a \ \fby \ c \ \fby \ b}{\rho} \rightsquigarrow^* \Conf{\noop}{\rho'}\) implies
  that \(\Conf{a \ \fby \ d \ \fby \ b}{\rho} \rightsquigarrow^* \Conf{\noop}{\rho'}\), and vice-versa.
\end{definition}

\begin{theorem}
    The operational semantics \(\,\rightsquigarrow^*\) for \(\StabLang\)
    is sound, complete, and universal for \(\StabChan_p\).
\end{theorem}
\begin{proof}
  Because the operational semantics are given by a deterministic
  transition system, given any configuration \(\Conf{t}{\rho}\), there is a
  unique \(\rho'\) such that \(\Conf{c}{\rho} \rightsquigarrow^* \Conf{\noop}{
  \rho'}\). Therefore, any well-formed term \(t\) yields a unique stabiliser
  quantum channel \(\Conf{t}{-}\). The observational equivalence of well-formed judgements \(c\) and \(d\)
  under \(\,\rightsquigarrow^*\) therefore amounts to equality as stabiliser
  quantum channels, \(\Conf{c}{-} = \Conf{d}{-}\), and thus, equality as quantum channels.
\end{proof}

\subsection{Denotational semantics}
We give \(\StabLang\) a denotational semantics in \(\Total(\ARQ)\). On types, let
\(\interp{\pittype} \coloneqq \F_p\) and \(\interp{\qpittype} \coloneqq
Q(\F_p^2, \omega_2)\). Define the denotation of a typed environment to be the dependent direct sum \(\interp{\Gamma} \coloneqq \bigoplus_{\reg{x} \in \dom(\Gamma)}
\interp{\Gamma(\reg{x})}\). 

The denotation of well-formed judgments \(\Gamma
\vdash t \update \Delta\) is given by the maps \(\ARQ(\interp{\Gamma},
\interp{\Delta})\) defined inductively from the denotation of generating terms.
As before, we need to establish notation to represent affine relations acting on a subset of the registers of the context.
Take ordered subsets \(\vbreg{x} \subseteq \dom(\Gamma)\) and \(\vbreg{y}\subseteq \dom(\Delta)\), where moreover, \(\dom(\Gamma) \setminus \vbreg{x} = \dom(\Delta) \setminus \vbreg{y}\). 
Given a relation \(S:\interp{\Gamma|_{\vbreg{x}}}\to \llbracket\Delta|_{\vbreg{y}}\rrbracket\) let \(S_{\vbreg{x}, \vbreg{y}}:\interp{\Gamma} \to \interp{ \Gamma' }\) denote the relation acting as \(S\) on the subspace \(\interp{\Gamma|_{\vbreg{x}}} \subseteq \interp{\Gamma}\) and trivially everywhere else \(\interp{\Gamma|_{\dom(\Gamma) \setminus \vbreg{x}}}\subseteq \interp{\Gamma}\).
The denotation of terms, defined inductively in figure~\ref{fig:denotation}.
\begin{figure}[t]
\centering
\setlength{\tabcolsep}{0pt}
\begin{tabular}{@{}l@{\;}l@{}} 
\textbf{Term} & \textbf{Denotation} \\ \hline

\(\interp{\noop}\) &
\(\coloneqq 1_{\interp{\Gamma}}\) \\

\(\interp{c \fby d}\) &
\(\coloneqq \interp{c}\,;\,\interp{d}\) \\

\(\interp{\discard\;\reg{u}}\) &
\(\coloneqq \Bigl\{ (u,\bullet)\in \F_p \oplus \{\bullet\} \Bigr\}_{\reg{u},\varnothing}\) \\

\(\interp{\pit\;\reg{u}}\) &
\(\coloneqq \Bigl\{ (\bullet,0)\in \{\bullet\}\oplus \F_p \Bigr\}_{\varnothing,\reg{u}}\) \\

\(\interp{\vbreg{v}=A\affmul\vbreg{u}}\) &
\(\coloneqq \Bigl\{ (\vb u,\vb v)\in \F_p^{m}\oplus \F_p^{n} \Bigm| \vb v=A(\vb u) \Bigr\}_{\vbreg{u},\vbreg{v}}\) \\

\(\interp{\qpit\;\reg{u}}\) &
\(\coloneqq \Bigl\{ \Bigl(\bullet,\bigl[\begin{smallmatrix}0\\ u_z\end{smallmatrix}\bigr]\Bigr)\in \{\bullet\}\oplus Q(\F_p^{2},\omega_{1}) \Bigr\}_{\varnothing,\reg{u}}\) \\

\(\interp{\measure\;\reg{u}}\) &
\(\coloneqq \Bigl\{ \Bigl(\bigl[\begin{smallmatrix}u_x\\ u_z\end{smallmatrix}\bigr],u_x\Bigr)\in Q(\F_p^{2},\omega_{1})\oplus \F_p \Bigr\}_{\reg{u},\reg{u}}\) \\

\(\interp{\reg{u}\muleqq X}\) &
\(\coloneqq \Bigl\{
\Bigl(
  \bigl[\begin{smallmatrix}u_x\\u_z\end{smallmatrix}\bigr],
  \bigl[\begin{smallmatrix}u_x \plus 1\\u_z\end{smallmatrix}\bigr]
\Bigr)
\Bigm|
\bigl[\begin{smallmatrix}u_x\\ u_z \end{smallmatrix}\bigr]\in Q(\F_p^2,\omega_1)
\Bigr\}_{\reg{u},\reg{u}}\) \\

\(\interp{\reg{u}\muleqq F}\) &
\(\coloneqq \Bigl\{
\Bigl(
  \bigl[\begin{smallmatrix}u_x\\u_z\end{smallmatrix}\bigr],
  \bigl[\begin{smallmatrix} u_z\\ \minu u_x\end{smallmatrix}\bigr]
\Bigr)
\Bigm|
\bigl[\begin{smallmatrix}u_x\\ u_z \end{smallmatrix}\bigr]\in Q(\F_p^2,\omega_1)
\Bigr\}_{\reg{u},\reg{u}}\) \\

\(\interp{\reg{u}\muleqq P}\) &
\(\coloneqq \Bigl\{
\Bigl(
  \bigl[\begin{smallmatrix}u_x\\u_z\end{smallmatrix}\bigr],
  \bigl[\begin{smallmatrix}u_x\\u_z-u_x\end{smallmatrix}\bigr]
\Bigr)
\Bigm| \bigl[\begin{smallmatrix}u_x\\ u_z \end{smallmatrix}\bigr]\in Q(\F_p^2,\omega_1)
\Bigr\}_{\reg{u},\reg{u}}\) \\

\(\interp{(\reg{u}, \reg{v})\muleqq \swap}\) &
\(\coloneqq \Bigl\{
\bigl(
 (\bigl[\begin{smallmatrix}u_x\\u_z\end{smallmatrix}\bigr],
  \bigl[\begin{smallmatrix}v_x\\v_z\end{smallmatrix}\bigr]),
 (\bigl[\begin{smallmatrix}v_x\\v_z\end{smallmatrix}\bigr],
  \bigl[\begin{smallmatrix}u_x\\u_z\end{smallmatrix}\bigr])
\bigr)
\Bigm|
\bigl[\begin{smallmatrix}u_x\\ u_z \end{smallmatrix}\bigr],
\bigl[\begin{smallmatrix}v_x\\ v_z \end{smallmatrix}\bigr]\in Q(\F_p^2,\omega_1)
\Bigr\}_{(\reg{u},\reg{v}),(\reg{u},\reg{v})}\) \\

\(\interp{(\reg{u}, \reg{v})\muleqq  C^X}\) &
\(\coloneqq \Bigl\{
\bigl(
 (\bigl[\begin{smallmatrix}u_x\\u_z\end{smallmatrix}\bigr],
  \bigl[\begin{smallmatrix}v_x\\v_z\end{smallmatrix}\bigr]),
 (\bigl[\begin{smallmatrix}u_x\\u_z -v_z\end{smallmatrix}\bigr],
  \bigl[\begin{smallmatrix}u_x\plus v_x\\v_z\end{smallmatrix}\bigr])
\bigr)
\Bigm|
\bigl[\begin{smallmatrix}u_x\\ u_z \end{smallmatrix}\bigr],
\bigl[\begin{smallmatrix}v_x\\ v_z \end{smallmatrix}\bigr]\in Q(\F_p^2,\omega_1)
\Bigr\}_{(\reg{u},\reg{v}),(\reg{u},\reg{v})}\) \\

\(\interp{\control_{U}\;\reg{u}\;\reg{v}}\) &
\(\coloneqq \Bigl\{ \Bigl((u,\bigl[\begin{smallmatrix}v_x\\ v_z\end{smallmatrix}\bigr]),\,(u,\bigl[\begin{smallmatrix}v_x\plus c\cdot p_x\\ v_z\plus c\cdot p_z\end{smallmatrix}\bigr])\Bigr)
\Bigm| u\in \F_p,\bigl[\begin{smallmatrix}v_x\\v_z\end{smallmatrix}\bigr]\in Q(\Zp^2,\omega_1)\Bigr\}_{(\reg{u},\reg{v}),(\reg{u},\reg{v})}\)

\end{tabular}

\caption{\textsf{\textbf{Denotation of \(\StabLang\) terms into \(\Total(\AR^Q)\)}}, where typed environments omitted for brevity. 
\(U=\pi(p_a,p_x,p_z)\in\mathcal{P}_p\) is a Pauli operator,
\(A:\Zp^{m}\to\Zp^n\) is an \(\Zp\)-affine transformation. For the sake of exposition we give the explicit interpretation for a generating set of the Clifford group: where \(X\) is the Pauli \(X\) gate,
\(C^X\) is the (quantum) controlled-\(X\) gate, \(F\) is the Fourier transform,
\(P\) is the quadratic phase gate.
}
\label{fig:denotation}
\end{figure}

\begin{theorem}[Full abstraction]
  \label{theorem:full_abstraction}
  Well-formed judgements \(c\) and \(d\) are observationally equivalent if and only
  if \(\interp{c} = \interp{d}\).
\end{theorem}

\begin{proof}
  Since all generating judgements are stabiliser, it follows from a
  straightforward induction that \(\Conf{c}{-}\) lies in the embedding of
  \(\StabChan_p\) into the category of CPTP maps between
  finite-dimensional \(C^*\)-algebras. By construction \(\interp{c}
  = \rel(\Conf{c}{-})\), therefore the claim follows from
  proposition~\ref{proposition:causal_total_split}.
\end{proof}

  It is conceptually straightforward, albeit notationally cumbersome,
  to move to a sampling operational semantics by decomposing measurements
  into projections along each outcome. The same denotational semantics is
  fully abstract for the observational equivalence of programs defined as
  statistical indistinguishability of configurations.

\subsection{Adding arbitrary classical operations to \texorpdfstring{\(\StabLang\)}{SPL}}
By enlarging the permissible classical operations of \(\StabLang\) with arbitrary classical operations, we can express a larger class of quantum channels:

\pagebreak
\begin{definition}
    Let \(\nlStabLang\) denote the  language given by adding an additional operation: \(\reg{z} = \mul\affmul  (\reg{x},\reg{y})\), alongside the formation rule:
    \vspace{-6mm}
    \begin{mathpar}
    \inferrule {\quad}{ \reg{x} : \pittype, \reg{y}:\pittype, \Gamma \vdash \reg{z} = \mul\!\star\! (\reg{x},\reg{y}) \update \reg{x} : \pittype, \reg{y} : \pittype, \reg{z} : \pittype, \Gamma}
  \end{mathpar}
\end{definition}

\begin{definition}
  We extend the small-step operational semantics of \(\StabLang\) to \(\nlStabLang\) with the additional rule given by:
  \begin{align*}
    \aden{\reg{z}\!=\!\mul\star(\reg{x},\reg{y})}{\reg{x}:\pittype,\reg{y}:\pittype,\Gamma}{\reg{x}:\pittype,\reg{y}:\pittype,\reg{z}:\pittype,\Gamma}
    \!\coloneqq\!
      \Big(
      &\rho \mapsto \sum_{j,k\in \Zp} \dyad{j,k,j\cdot k}{j,k} \rho \dyad{j,k}{j,k,j\cdot k} \\
      &:
      (\mathcal{H}_p,\mathcal{E}_Z)\otimes(\mathcal{H}_p,\mathcal{E}_Z)
      \!\to\!
      (\mathcal{H}_p,\mathcal{E}_Z)\Big)
      _{(\reg{x},\reg{y}),\reg{z}}
  \end{align*}
\end{definition}

By construction it is immediate that:

\begin{lemma}
    The operational semantics \(\,\rightsquigarrow^*\) for \(\nlStabLang\) is sound, complete, and universal for the observational equivalence of stabiliser quantum channels with arbitrary classical operations.
\end{lemma}

On the other hand, \(\nlStabLang\) also admits a denotational semantics:

\begin{lemma}
  The denotational semantics of \(\StabLang\) into \(\Total(\ARQ)\) extends to a denotational semantics of \(\nlStabLang\) into \(\NLQ\) given by:
  \begin{center}
  \(\interp{\reg{z}\!=\!\mul\!\star\!(\reg{x}, \reg{y})} \coloneqq \Big\{ \big( \bigl[\begin{smallmatrix} a \\ b \end{smallmatrix}\bigr], a\cdot b \big) \in \Zp^2 \oplus \Zp\Big\}_{(\reg{x},\reg{y}), \reg{z}}:\Zp \oplus \Zp\to \Zp\)
  \end{center}
\end{lemma}

We weaken observational equivalence to forget about the probability distribution of measurements, recording only the \emph{support} of measurements:

\begin{definition}
    Say that two quantum channels are \textbf{nondeterministically observationally} equivalent in case they produce the same \textbf{possible} measurement outcomes according to the Born rule when acting on arbitrary density matrices.
\end{definition}

This motivates the following conjecture:

\begin{conjecture}
Well-formed judgements \(c\) and \(d\) in \(\nlStabLang\) are nondeterministically observationally equivalent if and only  if \(\interp{c} = \interp{d}\).
\end{conjecture}

This seems nontrivial to prove; however, it is well-motivated. Morphisms in \(\NLQ\) are relations, therefore they record the possible ways in which stabilisers are related to each other.  The question is if the nonlinearity of classical operations interacts well with the composition of quantum channels. Despite the breadth of research on stabiliser quantum mechanics, adding classically controlled Pauli operations dependent non-affinely on Pauli measurements is more complex than one might na\"ively assume.

\section{Conclusion}

We have developed a denotational semantics for stabiliser quantum programs
which allows for the manipulation of stabiliser codes, Pauli measurements,
and affine classical operations and classically controlled Pauli
operators. We demonstrated this semantics by giving a fully abstract
denotational semantics to a toy imperative stabiliser language. Furthermore,
we extended this language with arbitrary classical control, with a corresponding
denotational semantics. We conjecture that this is fully abstract with respect
to the equivalence relation induced by \emph{possible} measurement outcomes.

In the case of qubits, the affine, symplectic representation of stabiliser maps breaks down so that \(\Proj(\Stab_2) \not\simeq
\ALR[\F_2]\) \cite{comfort_graphical_2021}.
By restricting the unitary operations to be generated by the controlled-not gate, the Pauli group and the swap gate, we obtain the maximal subcategory of qubit stabiliser maps on which the symplectic representation still holds \cite[p. 156]{comfortthesis}.  This is the natural setting for CSS codes \cite{Calderbank,Steane}, which are widely used in QEC \cite{eczoo_qubit_css}.  

The language we have described in this paper is quite low-level; despite the abstract geometric structure of its denotational semantics. In future work, we intend to develop a higher level programming language with primitives reflecting the elegant structure of the semantics. For example, the ability to natively represent graph states and graph-like operations, correctable and detectable errors, and the ability to make use of the enrichment in partially ordered sets  would be very useful.

It is also future work to explore denotational semantics for stabiliser quantum programs using their graphical calculus. There is a complete ZX-calculus for affine Lagrangian relations \cite{booth_graphical_2024}, which is equivalent to the qupit ZX-calculus \cite{quopitzx,Por2023} modulo scalars. This is an interesting direction for future work because the ZX-calculus has already been successful for constructing fault-tolerant quantum circuits \cite{bombin}, and the design and verification of QEC codes \cite{Chancellor,Duncan2014}.

\bibliography{bibliography}

\appendix

\end{document}